\documentclass[final,3p]{elsarticle}
\usepackage{lineno,hyperref}
\usepackage{amsmath}
\usepackage{amssymb}
\usepackage{algorithmic}
\usepackage{algorithm}
\usepackage{caption}
\usepackage{mathtools}
\usepackage{changes}
\usepackage{multirow}
\usepackage{verbatim}
\usepackage{mathrsfs}
\usepackage{graphicx}
\usepackage{amsmath,amsthm,bm,mathrsfs}

\theoremstyle{plain}
\newtheorem{theorem}{Theorem}[section]

\theoremstyle{definition}

\theoremstyle{remark}

\modulolinenumbers[5]

\journal{ArXiv.org}

\begin{document}

\begin{frontmatter}

\title{Compression and Distillation of Data Quadruplets in Non-intrusive Reduced-order Modeling}

\author[uz]{Umair~Zulfiqar}
\author[qs]{Qiu-Yan~Song\corref{mycorrespondingauthor}}
\cortext[mycorrespondingauthor]{Corresponding author}
\ead{qysong@shu.edu.cn}
\author[zx]{Zhi-Hua~Xiao}
\author[ld1,ld2]{Liang~Du}
\author[vs]{Victor~Sreeram}
\address[uz]{School of Electronic Information and Electrical Engineering, Yangtze University, Jingzhou, Hubei, 434023, China}
\address[qs]{Huangshi Key Laboratory of Metaverse and Virtual Simulation, School of Mathematics and Statistics, Hubei Normal University, Huangshi, Hubei 435002, China}
\address[zx]{School of Mathematics and Systems Science, Wuhan University of Science and Technology, Wuhan, Hubei 430065, China}
\address[ld1]{School of Computer and Information Technology, Shanxi University, Taiyuan 030006, China}
\address[ld2]{Key Laboratory of Evolutionary Science Intelligence of Shanxi Province, Shanxi University, Taiyuan 030006, Shanxi, China}
\address[vs]{Department of Electrical, Electronic, and Computer Engineering, The University of Western Australia, Perth, 6009, Australia}
\begin{abstract}
This paper introduces a quadrature-free, non-intrusive approach to balanced truncation. The method non-intrusively constructs reduced-order models using available transfer function samples from the right half of the $s$-plane. It is highlighted that the proposed non-intrusive balanced truncation and existing quadrature-based balanced truncation algorithms share a common feature: both compress their respective data quadruplets to derive reduced-order models. Additionally, it is shown that by using different compression strategies, these quadruplets can be used to develop three non-intrusive formulations of the Iterative Rational Krylov Algorithm (IRKA). These formulations non-intrusively generate reduced models using transfer function samples from the $j\omega$-axis or the right half of the $s$-plane, or impulse response samples. Notably, these IRKA formulations eliminate the necessity of computing new transfer function samples as IRKA iteratively updates the interpolation points. The non-intrusive algorithms developed in this paper are also extended to discrete-time systems. The efficacy of the proposed algorithms is validated through numerical examples, which show that the proposed non-intrusive approaches perform comparably to their intrusive counterparts.
\end{abstract}


\end{frontmatter}

\section{Introduction}
Model order reduction (MOR) comprises system-theoretic methods aimed at constructing simplified models that accurately replicate the input-output behavior of large-scale dynamical systems. By efficiently capturing key dynamical characteristics of the original system, reduced order models (ROMs) are able to approximate its behavior across a broad range of inputs, yet are significantly lower in order. These ROMs are designed to be computationally efficient, making them easier to simulate, manipulate, and control. For further details on various MOR techniques, the readers are referred to \cite{antoulas2017tutorial,obinata2012model,antoulas2005approximation}.

Balanced truncation (BT) \cite{moore1981principal} is a highly effective and widely used technique for MOR of linear dynamical systems. This method preserves the asymptotic stability of the original system while offering \textit{a priori} error bounds for the ROM. By discarding states that are difficult to reach and observe, as determined by the relative magnitude of the system's Hankel singular values, BT ensures that their impact on the system's input-output behavior is minimal. Consequently, the ROM accurately approximates the original system in simulations or analyses.

The primary computational burden in BT lies in solving large-scale Lyapunov equations to compute the system Gramians. Various approaches, such as those mentioned in the surveys \cite{benner2013numerical,simoncini2016computational}, have been developed to efficiently compute these Gramians. These methods rely on the system's explicit state-space representation, making BT an ``\textit{intrusive}'' method. This is in contrast to ``\textit{non-intrusive}'' methods, which depend solely on system response data—like transfer function samples or impulse response measurements—without requiring the system's internal state-space representation \cite{mayo2007framework,nakatsukasa2018aaa,gosea2022data,scarciotti2024interconnection}. In \cite{goseaQuad}, a non-intrusive BT algorithm based on numerical integration, called Quadrature-BT (QuadBT), is introduced. This algorithm constructs the ROM using transfer function samples at the $j\omega$-axis of the $s$-plane or samples of impulse response and its derivatives.

The $\mathcal{H}_2$-optimal MOR problem involves finding a local minimum for the (squared) $\mathcal{H}_2$ norm of the error transfer function. One of the key methods for achieving this local optimum is the Iterative Rational Krylov Algorithm (IRKA) \cite{gugercin2008h_2}. A non-intrusive version of IRKA was introduced in \cite{beattie2012realization}, based on the interpolatory framework proposed in \cite{mayo2007framework}. This approach requires only transfer function samples and their derivatives to compute the local optimum, making it non-intrusive. However, because IRKA is iterative, the sampling points are updated at each iteration and cannot be predetermined. Instead, IRKA identifies the optimal sampling points through successive iterations. If these new samples must be estimated experimentally, the algorithm must pause until the new approximations of samples are available. This poses practical challenges, as it may be difficult or even impossible to conduct experiments to estimate transfer function samples each time the sampling points are updated.

Among the various $\mathcal{H}_2$ MOR algorithms, the Pseudo-optimal Rational Krylov (PORK) algorithm is an important suboptimal method \cite{wolf2013h,wolf2014h}. Unlike IRKA, PORK is an iteration-free approach that satisfies a subset of the $\mathcal{H}_2$ optimality conditions in a single run. In this paper, PORK plays a significant role in the development of the non-intrusive implementations of both BT and IRKA.

Over the past two decades, the low-rank Alternating-direction Implicit (ADI) method has proven highly effective in reducing the computational cost of BT \cite{benner2013efficient}. It is now one of the most widely used and efficient BT algorithms in the literature \cite{saak2012goal}. In this paper, we introduce a non-intrusive implementation of the low-rank ADI-based BT that constructs the ROM from transfer function samples in the right-half of the \(s\)-plane. Unlike QuadBT, this approach does not rely on numerical integration. Additionally, we propose three non-intrusive implementations for IRKA, tailored to the type of data available. In cases where transfer function samples along the \(j\omega\) axis or impulse response measurements are accessible, we present numerical integration-based algorithms that do not require new transfer function samples as IRKA updates the sampling points. For scenarios where transfer function samples in the right-half of the \(s\)-plane are available, we propose a version that does not require numerical integration and new transfer function samples as IRKA updates the sampling points. Additionally, all these non-intrusive implementations (both BT and IRKA) for continuous-time systems are extended to discrete-time systems in this paper.

The remainder of the paper is structured as follows. Section \ref{sec2} provides the necessary background on MOR and briefly reviews existing MOR algorithms most relevant to this work. The main contributions of this research begin in Section \ref{sec3}, where a non-intrusive implementation of ADI-based low-rank BT is proposed. Section \ref{sec4} presents three new non-intrusive implementations of IRKA, tailored to the type of available data. Section \ref{sec5} introduces two quadrature-based non-intruisve implementations of IRKA for discrete-time systems. In Section \ref{sec6}, the PORK algorithm is extended to discrete-time systems. Building on this, Section \ref{sec7} formulates a quadrature-free non-intrusive implementation of BT for discrete-time systems, while Section \ref{sec8} develops a quadrature-free non-intrusive implementation of IRKA. Section \ref{sec9} elaborates on the concepts of compression and distillation in the context of non-intrusive MOR. The performance of the proposed algorithms is evaluated in Section \ref{sec10}. Finally, the paper concludes in Section \ref{sec11}.
\section{Preliminaries}\label{sec2}
Consider an \( n^{th} \)-order linear time-invariant (LTI) system \( G(s) \) represented by the state-space realization  
\[
G(s) = C(sE - A)^{-1}B,
\]  
where \( E \in \mathbb{R}^{n \times n} \), \( A \in \mathbb{R}^{n \times n} \), \( B \in \mathbb{R}^{n \times m} \), and \( C \in \mathbb{R}^{p \times n} \). 

Suppose the \( r^{th} \)-order ROM \( G_r(s) \) is given by the state-space realization  
\[
G_r(s) = C_r(sE_r - A_r)^{-1}B_r,
\]  
where \( E_r \in \mathbb{R}^{r \times r} \), \( A_r \in \mathbb{R}^{r \times r} \), \( B_r \in \mathbb{R}^{r \times m} \), and \( C_r \in \mathbb{R}^{p \times r} \). Throughout the paper, $G(s)$ and $G_r(s)$ are assumed to be stable while the matrices $E$ and $E_r$ are assumed to be non-singular.

The ROM is derived from \( G(s) \) using Petrov-Galerkin projection, defined as  
\[
E_r = W^T E V, \quad A_r = W^T A V, \quad B_r = W^T B, \quad C_r = C V,
\]  
where \( W \in \mathbb{R}^{n \times r} \), \( V \in \mathbb{R}^{n \times r} \), and both \( V \) and \( W \) are full column rank matrices. Let $T_v \in \mathbb{C}^{r \times r}$ and $T_w \in \mathbb{C}^{r \times r}$ be invertible matrices. The projection matrices \( W \) and \( V \) can be substituted with \( WT_w \) and \( VT_v \), yielding the same ROM \( G_r(s) \) but with a different state-space realization. This property can be utilized to transform complex projection matrices and the resulting state-space matrices of the ROM into real-valued ones. For the sake of clarity and simplicity in presentation, we will assume \( V \), \( W \), \( E_r \), \( A_r \), \( B_r \), and \( C_r \) to be complex matrices throughout the remainder of the paper, without any loss of generality. Readers are referred to (Section 4.1 of) \cite{goseaQuad} for computing \(T_v\) and \(T_w\) to ensure that the ROMs obtained using the algorithms discussed in the following sections are real-valued.
\subsection{Review of Interpolation Theory \cite{beattie2017model}}
Let the right interpolation points be \((\sigma_1,\dots, \sigma_r)\) and the left interpolation points be \((\mu_1,, \dots, \mu_r)\), with their corresponding right tangential directions \((b_1, \dots, b_r)\) and left tangential directions \((c_1, \dots, c_r)\). The projection matrices \(V \in \mathbb{C}^{n \times r}\) and \(W \in \mathbb{C}^{n \times r}\) within the interpolation framework can be constructed as follows:
\begin{align}
V &= \begin{bmatrix} (\sigma_1 E - A)^{-1} B b_1 & \cdots & (\sigma_r E - A)^{-1} B b_r \end{bmatrix},\label{Kry_V}\\
W &= \begin{bmatrix} (\mu_1^* E^T - A^T)^{-1} C^T c_1^* & \cdots & (\mu_r^* E^T - A^T)^{-1} C^T c_r^* \end{bmatrix},\label{Kry_W}
\end{align}where $b_i\in\mathbb{C}^{m\times 1}$ and $c_i\in\mathbb{C}^{1\times p}$.

The ROM obtained using these projection matrices satisfies the following tangential interpolation conditions:
\begin{align}
G(\sigma_j) b_j = G_r(\sigma_j) b_j, \quad c_i G(\mu_i) = c_i G_r(\mu_i),\label{int_cond_1}
\end{align}
for \(i = 1, \dots, r\) and \(j = 1, \dots, r\). Additionally, if there are common right and left interpolation points, i.e., \(\sigma_j = \mu_i\), the following tangential Hermite interpolation conditions are also satisfied for those points:
\begin{align}
c_i G^\prime(\sigma_j) b_j = c_i G_r^\prime(\sigma_j) b_j.\label{int_cond_2}
\end{align}
\subsection{Iterative Rational Krylov Algorithm (IRKA) \cite{gugercin2008h_2}}
Assume that \( G(s) \) and \( G_r(s) \) have simple poles. In this case, they can be expressed in the following pole-residue form:
\[
G(s) = \sum_{k=1}^{n} \frac{l_k r_k^*}{s - \lambda_k}, \quad G_r(s) = \sum_{k=1}^{r} \frac{\hat{l}_k \hat{r}_k^*}{s - \hat{\lambda}_k}.
\]
The necessary conditions for a local optimum of \( ||G(s) - G_r(s)||_{\mathcal{H}_2}^2 \) are given by:
\begin{align}
\hat{l}_i^*G^{\prime}(-\hat{\lambda}_i)\hat{r}_i&=\hat{l}_i^*G_r^{\prime}(-\hat{\lambda}_i)\hat{r}_i,\label{op1}\\
\hat{l}_i^*G(-\hat{\lambda}_i)&=\hat{l}_i^*G_r(-\hat{\lambda}_i),\label{op2}\\
G(-\hat{\lambda}_i)\hat{r}_i&=G_r(-\hat{\lambda}_i)\hat{r}_i,\label{op3}
\end{align}for $i=1,2,\cdots,r$.

Since the ROM \( G_r(s) \) is initially unknown, IRKA uses fixed-point iterations starting from an arbitrary initial guess of the interpolation data to search for the local optimum. After each iteration, the interpolation data is updated as \( \sigma_i = \mu_i = -\hat{\lambda}_i \), \( b_i = \hat{r}_i \), and \( c_i = \hat{l}_i^* \) until convergence is achieved. Upon convergence, a local optimum of \( ||G(s) - G_r(s)||_{\mathcal{H}_2}^2 \) is achieved.
\subsection{Pseudo-optimal Rational Krylov (PORK) Algorithm \cite{wolf2014h}}
Let us define $S_b$, $S_c$, $L_b$, and $L_c$ as follows:
\begin{align}
S_b &= \text{diag}(\sigma_1, \dots, \sigma_r),& S_c &= \text{diag}(\mu_1, \dots, \mu_r),\nonumber\\
L_b &= \begin{bmatrix} b_1, \dots, b_r \end{bmatrix},& L_c^* &= \begin{bmatrix} c_1^*, \dots, c_r^* \end{bmatrix}.\label{SbLbScLc}
\end{align}
The projection matrices \( V \) and \( W \) in (\ref{Kry_V}) and (\ref{Kry_W}), respectively, solve the following Sylvester equations:
\begin{align}
AV-EVS_b+BL_b&=0,\label{sylv_V}\\
A^TW-E^TWS_c^*+C^TL_c^*&=0.\label{sylv_W}
\end{align}By pre-multiplying (\ref{sylv_V}) with \( W^* \), it can be observed that the matrix \( A_r \) can be expressed as \( A_r = E_r S_b - B_r L_b \). This allows \( A_r \) to be parameterized in terms of \( E_r \) and \( B_r \) without affecting the interpolation conditions induced by \( V \), as this is equivalent to varying \( W \). Assume the pair \( (-S_b, L_b) \) is observable and solves the following Lyapunov equation:
\begin{align}
-S_b^* Q_s - Q_s S_b + L_b^* L_b = 0.\label{pork_qs}
\end{align}
By setting \( E_r = I \) and \( B_r = Q_s^{-1} L_b^* \), \( A_r \) becomes \( A_r = -Q_s^{-1} S_b^* Q_s \). The resulting ROM:
\begin{align}
E_r &= I, & A_r &= -Q_s^{-1} S_b^* Q_s,\nonumber\\
B_r &= Q_s^{-1} L_b^*, & C_r &= CV,\nonumber
\end{align}
satisfies the optimality condition (\ref{op3}). This approach will be referred to as Input PORK (I-PORK) throughout this paper.

Similarly, by pre-multiplying (\ref{sylv_W}) with $V^*$, it can be noted that \( A_r \) can also be represented as \( A_r = S_c E_r - L_c C_r \). This allows \( A_r \) to be parameterized in terms of \( E_r \) and \( C_r \) without affecting the interpolation conditions induced by \( W \), as this is equivalent to varying \( V \). Assume the pair \( (-S_c, L_c) \) is controllable and solves the following Lyapunov equation:
\begin{align}
-S_c P_s - P_s S_c^* + L_c L_c^* = 0.\label{pork_ps}
\end{align}
By setting \( E_r = I \) and \( C_r = L_c^* P_s^{-1} \), \( A_r \) becomes \( A_r = -P_s S_c^* P_s^{-1} \). The resulting ROM:
\begin{align}
E_r &= I, & A_r &= -P_s S_c^* P_s^{-1},\nonumber\\
B_r &= W^* B, & C_r &= L_c^* P_s^{-1},\nonumber
\end{align}
satisfies the optimality condition (\ref{op2}). This approach will be referred to as Output PORK (O-PORK) throughout this paper.
\subsection{Interpolatory Loewner framework \cite{mayo2007framework}}
In the Loewner framework, the matrices of the ROM, which satisfies the interpolation condition (\ref{int_cond_1}), are constructed from transfer function samples at the interpolation points as follows:
\begin{align}
W^*EV&=\begin{bmatrix}-\frac{c_1G(\sigma_1)b_1 - c_1G(\mu_1)b_1}{\sigma_1 - \mu_1}&\cdots&-\frac{c_1G(\sigma_r)b_r - c_1G(\mu_1)b_r}{\sigma_r - \mu_1}\\\vdots&\ddots&\vdots\\-\frac{c_rG(\sigma_1)b_1 - c_rG(\mu_r)b_1}{\sigma_1 - \mu_r}&\cdots&-\frac{c_rG(\sigma_r) - G(\mu_r)b_r}{\sigma_r - \mu_r}\end{bmatrix},\nonumber\\
W^*AV&=\begin{bmatrix}-\frac{\sigma_1 c_1G(\sigma_1)b_1 - \mu_1 c_1G(\mu_1)b_1}{\sigma_1 - \mu_1}&\cdots&-\frac{\sigma_r c_1G(\sigma_r)b_r - \mu_1 c_1G(\mu_1)b_r}{\sigma_r - \mu_1}\\\vdots&\ddots&\vdots\\-\frac{\sigma_1 c_rG(\sigma_1)b_1 - \mu_r c_rG(\mu_r)b_1}{\sigma_1 - \mu_r}&\cdots&-\frac{\sigma_r c_rG(\sigma_r)b_r - \mu_r c_rG(\mu_r)b_r}{\sigma_r - \mu_r}\end{bmatrix},\nonumber\\
W^*B&=\begin{bmatrix}c_1G(\mu_1)\\\vdots\\c_rG(\mu_r)\end{bmatrix},\quad CV=\begin{bmatrix}G(\sigma_1)b_1&\cdots G(\sigma_r)b_r\end{bmatrix},\label{LF}
\end{align}where $V$ and $W$ are as in (\ref{Kry_V}) and (\ref{Kry_W}), respectively.
When \( \sigma_j \approx \mu_i \), the expressions approach to:
\begin{align}
 \frac{c_iG(\sigma_j)b_j - c_iG(\mu_i)b_j}{\sigma_j - \mu_i}&\approx c_i G^\prime(\sigma_j) b_j,\nonumber\\
  \frac{\sigma_j c_iG(\sigma_j)b_j - \mu_i c_iG(\mu_i)b_j}{\sigma_j - \mu_i} &\approx c_iG(\sigma_j)b_j + \sigma_j c_iG^\prime(\sigma_j) b_j.\nonumber
\end{align}Thus, when there are common elements in the sets of right and left interpolation points, samples of the derivative of \( G(s) \) at those common points are also required to construct \( W^*EV \) and \( W^*AV \). If block interpolation is required instead of tangential interpolation, one can treat \( b_j \) and \( c_i \) as scalars and set them to \( b_j = c_i = 1 \) in the formulas above.

The matrices \( E_r \) and \( A_r \) in the above formulas exhibit a special structure known as the Loewner matrix and shifted Loewner matrix, respectively. This structure is the reason behind the name ``Interpolatory Loewner framework''.
\subsection{Balanced Truncation (BT) \cite{moore1981principal}}
Let \( P \) and \( Q \) denote the controllability and observability Gramians, respectively, defined by the following integral expressions:
\begin{align}
P &= \frac{1}{2\pi} \int_{-\infty}^{\infty} (j\omega E - A)^{-1} BB^T (-j\omega E^T - A^T)^{-1} \, d\omega, \label{int1}\\
Q &= \frac{1}{2\pi} \int_{-\infty}^{\infty} (-j\omega E^T - A^T)^{-1} C^T C (j\omega E - A)^{-1} \, d\omega. \label{int2}
\end{align} \( P \) and \( Q \) can also be expressed using time-domain integral formulas as follows:
\begin{align}
P&=\int_{0}^{\infty}e^{E^{-1}A\tau}E^{-1}BB^TE^{-T}e^{A^TE^{-T}\tau}d\tau,\label{int3}\\
Q&=\int_{0}^{\infty}e^{E^{-T}A^T\tau}E^{-T}C^TCE^{-1}e^{AE^{-1}\tau}d\tau.\label{int4}
\end{align}
The Gramians \( P \) and \( Q \) can be computed by solving the following Lyapunov equations:
\begin{align}
APE^T + EPA^T + BB^T = 0,\label{lyap_P}\\
A^TQE + E^TQA + C^T C = 0.\label{lyap_Q}
\end{align}
Next, compute the Cholesky factorizations of \( P \) and \( Q \) as:
\[
P = Z_p Z_p^T \quad \text{and} \quad Q = Z_q Z_q^T.
\]
The balancing square-root algorithm \cite{tombs1987truncated} proceeds as follows. First, compute the singular value decomposition (SVD) of \( Z_q^T E Z_p \):
\[
Z_q^T E Z_p = \begin{bmatrix} U_1 & U_2 \end{bmatrix} \begin{bmatrix} S_1 & 0 \\ 0 & S_2 \end{bmatrix} \begin{bmatrix} V_1^T \\ V_2^T \end{bmatrix}.
\]
Finally, the projection matrices \( W \) and \( V \) in BT are constructed as:
\[
W = Z_q U_1 S_1^{-\frac{1}{2}} \quad \text{and} \quad V = Z_p V_1 S_1^{-\frac{1}{2}}.
\]
\subsection{Non-intrusive Quadrature-based Balanced Truncation (QuadBT)\cite{goseaQuad}\label{subsec_QuadBT}}
Our presentation of QuadBT differs slightly from the original formulation in \cite{goseaQuad}. This choice of presentation aims to emphasize that QuadBT, like all the algorithms proposed in this paper, compresses and distills data quadruplets to construct the ROM. The concepts of compression and distillation in the context of non-intrusive MOR will be discussed in detail in Section \ref{sec9}.

The integrals (\ref{int1}) and (\ref{int2}) can be approximated using a numerical quadrature rule as follows:
\begin{align}
P &\approx \hat{P} = \sum_{i=1}^{n_p} w_{p,i}^2 (j\omega_i E - A)^{-1} B B^T (-j\omega_i E^T - A^T)^{-1} + w_{p,\infty}^2 E^{-1} B B^T E^{-T},\nonumber\\
Q &\approx \hat{Q} = \sum_{i=1}^{n_q} w_{q,i}^2 (-j\nu_i E^T - A^T)^{-1} C^T C (j\nu_i E - A)^{-1} + w_{q,\infty}^2 E^{-T} C^T C E^{-1},\nonumber
\end{align}
where \(\omega_i\) and \(\nu_i\) are the quadrature nodes, and \(w_{p,i}^2\) and \(w_{q,i}^2\) are the corresponding quadrature weights. The weights \(w_{p,\infty}^2\) and \(w_{q,\infty}^2\) are associated with the nodes at infinity. The low-rank factors of \(P\) and \(Q\), denoted as \(\hat{P} = \hat{Z}_p \hat{Z}_p^T\) and \(\hat{Q} = \hat{Z}_q \hat{Z}_q^T\), can be decomposed as:
\[
\hat{Z}_p = \tilde{V} L_p, \quad \hat{Z}_q = \tilde{W} L_q,
\]
where
\begin{align}
\tilde{V} &= \begin{bmatrix} (j\omega_1 E - A)^{-1} B & \cdots & (j\omega_{n_p} E - A)^{-1} B & E^{-1} B \end{bmatrix},\label{V_tild}\\
\tilde{W} &= \begin{bmatrix} (-j\nu_1 E^T - A^T)^{-1} C^T & \cdots & (-j\nu_{n_q} E^T - A^T)^{-1} C^T & E^{-T} C^T \end{bmatrix},\label{W_tild}\\
L_p &=   \text{diag}(w_{p,1}, \dots, w_{p,n_p}, w_{p,\infty})\otimes I_m,\nonumber\\
L_q &=   \text{diag}(w_{q,1}, \dots, w_{q,n_q}, w_{q,\infty})\otimes I_p.\nonumber
\end{align}
The matrices \(L_p\) and \(L_q\) can be computed solely from the quadrature weights. Additionally, the terms \(E_w=\tilde{W}^* E \tilde{V}\), \(A_w=\tilde{W}^* A \tilde{V}\), \(B_w=\tilde{W}^* B\), and \(C_w=C \tilde{V}\) can be constructed non-intrusively using transfer function samples at the quadrature nodes within the Loewner framework as follows:
\begin{align}
E_w&=\begin{bmatrix}-\frac{G(j\omega_1) - G(j\nu_1)}{j\omega_1 - j\nu_1}&\cdots&-\frac{G(j\omega_{n_p}) - G(j\nu_1)}{j\omega_{n_p} - j\nu_1}\\\vdots&\ddots&\vdots\\-\frac{G(j\omega_1) - G(j\nu_{n_q})}{j\omega_1 - j\nu_{n_q}}&\cdots&-\frac{G(j\omega_{n_p}) - G(j\nu_{n_q})}{j\omega_{n_p} -j\nu_{n_q}}\end{bmatrix},\nonumber\\
A_w&=\begin{bmatrix}-\frac{j\omega_1 G(j\omega_1) - j\nu_1 G(j\nu_1)}{j\omega_1 - j\nu_1}&\cdots&-\frac{j\omega_{n_p} G(j\omega_{n_p}) - j\nu_1 G(j\nu_1)}{j\omega_{n_p} - j\nu_1}\\\vdots&\ddots&\vdots\\-\frac{j\omega_1 G(j\omega_1) - j\nu_{n_q} G(j\nu_{n_q})}{j\omega_1 - j\nu_{n_q}}&\cdots&-\frac{j\omega_{n_p} G(j\omega_{n_p}) - j\nu_{n_q} G(j\nu_{n_q})}{j\omega_{n_p} - j\nu_{n_q}}\end{bmatrix},\nonumber\\
B_w&=\begin{bmatrix}G(j\nu_1)\\\vdots\\G(j\nu_{n_q})\end{bmatrix},\quad C_w=\begin{bmatrix}G(j\omega_1)&\cdots& G(j\omega_{n_p})\end{bmatrix}.\label{jw_LF}
\end{align}
The low-rank factors \(\hat{Z}_p\) and \(\hat{Z}_q\) can then replace \(Z_p\) and \(Z_q\) in the balanced square root algorithm as:
\[
L_q^TE_wL_p = \begin{bmatrix} \hat{U}_1 & \hat{U}_2 \end{bmatrix} \begin{bmatrix} \hat{S}_1 & 0 \\ 0 & \hat{S}_2 \end{bmatrix} \begin{bmatrix} \hat{V}_1^* \\ \hat{V}_2^* \end{bmatrix}.
\]
Further, let the projection matrices \( \hat{W}_r \) and \( \hat{V}_r \) be defined as follows:
\[
\hat{W}_r = L_q \hat{U}_1 \hat{S}_1^{-\frac{1}{2}} \quad \text{and} \quad \hat{V}_r = L_p \hat{V}_1 \hat{S}_1^{-\frac{1}{2}}.
\]
The ROM in frequency-domain QuadBT is constructed by reducing the Loewner quadruplet \((E_w, A_w, B_w, C_w)\) as follows:
\begin{align}
E_r &= \hat{W}_r^* E_w \hat{V}_r=I, & A_r &= \hat{W}_r^* A_w \hat{V}_r, &B_r &= \hat{W}_r^* B_w, & C_r &= C_w \hat{V}_r.\nonumber
\end{align}
Similarly, the integrals (\ref{int3}) and (\ref{int4}) can be approximated using numerical quadrature as follows:
\begin{align}
P &\approx \sum_{i=1}^{n_p} w_{p,i}^2 e^{E^{-1} A t_i} E^{-1} B B^T E^{-T} e^{A^T E^{-T} t_i},\nonumber\\
Q &\approx \sum_{i=1}^{n_q} w_{q,i}^2 e^{E^{-T} A^T \tau_i} E^{-T} C^T C E^{-1} e^{A E^{-1} \tau_i}.\nonumber
\end{align}
The low-rank factors of \( P \) and \( Q \), denoted as \( \hat{P} = \hat{Z}_p \hat{Z}_p^T \) and \( \hat{Q} = \hat{Z}_q \hat{Z}_q^T \), can be decomposed as \( \hat{Z}_p = \tilde{V} L_p \) and \( \hat{Z}_q = \tilde{W} L_q \), where
\begin{align}
\tilde{V} &= \begin{bmatrix} e^{E^{-1} A t_1} E^{-1} B & \cdots & e^{E^{-1} A t_{n_p}} E^{-1} B \end{bmatrix},\label{Vt_tild}\\
\tilde{W} &= \begin{bmatrix} e^{E^{-T} A^T \tau_1} E^{-T} C^T & \cdots & e^{E^{-T} A^T \tau_{n_q}} E^{-T} C^T \end{bmatrix},\label{Wt_tild}\\
L_p &=   \text{diag}(w_{p,1}, \dots, w_{p,n_p}, w_{p,\infty})\otimes I_m,\nonumber\\
L_q &=   \text{diag}(w_{q,1}, \dots, w_{q,n_q}, w_{q,\infty})\otimes I_p.\nonumber
\end{align}
Let \( h(t) \) denote the impulse response of \( G(s) \). The impulse response and its derivative can be expressed as:
\begin{align}
h(t) &= C e^{E^{-1} A t} E^{-1} B = C E^{-1} e^{A E^{-1} t} B,\nonumber\\
h^\prime(t) &= C e^{E^{-1} A t} E^{-1} A E^{-1} B.\nonumber
\end{align}
The terms \( E_t=\tilde{W}^T E \tilde{V} \), \( A_t=\tilde{W}^T A \tilde{V} \), \( B_t=\tilde{W}^T B \), and \( C_t=C \tilde{V} \) can be constructed non-intrusively using samples of the impulse response and its derivative as follows:
\begin{align}
E_t &= \begin{bmatrix}
h(\tau_1 + t_1) & \cdots & h(\tau_1 + t_{n_p}) \\
\vdots & \ddots & \vdots \\
h(\tau_{n_q} + t_1) & \cdots & h(\tau_{n_q} + t_{n_p})
\end{bmatrix},\nonumber\\
A_t& = \begin{bmatrix}
h^\prime(\tau_1 + t_1) & \cdots & h^\prime(\tau_1 + t_{n_p}) \\
\vdots & \ddots & \vdots \\
h^\prime(\tau_{n_q} + t_1) & \cdots & h^\prime(\tau_{n_q} + t_{n_p})
\end{bmatrix},\nonumber\\
B_t &= \begin{bmatrix} h(\tau_1) \\ \vdots \\ h(\tau_{n_q}) \end{bmatrix}, \quad C_t = \begin{bmatrix} h(t_1) & \cdots & h(t_{n_p}) \end{bmatrix}.\label{impulse_quad}
\end{align}
Additionally, \( L_p \) and \( L_q \) can be computed from the quadrature weights. The low-rank factors \(\hat{Z}_p\) and \(\hat{Z}_q\) can then replace \(Z_p\) and \(Z_q\) in the balanced square root algorithm as:
\[
L_q^T E_t L_p = \begin{bmatrix} \hat{U}_1 & \hat{U}_2 \end{bmatrix} \begin{bmatrix} \hat{S}_1 & 0 \\ 0 & \hat{S}_2 \end{bmatrix} \begin{bmatrix} \hat{V}_1^* \\ \hat{V}_2^* \end{bmatrix}.
\]
Further, let the projection matrices \( \hat{W}_r \) and \( \hat{V}_r \) be defined as follows:
\[
\hat{W}_r = L_q \hat{U}_1 \hat{S}_1^{-\frac{1}{2}} \quad \text{and} \quad \hat{V}_r = L_p \hat{V}_1 \hat{S}_1^{-\frac{1}{2}}.
\]
The ROM in time-domain QuadBT is constructed by reducing the impulse data quadruplet \((E_t, A_t, B_t, C_t)\) as follows:
\begin{align}
E_r &= \hat{W}_r^* E_t \hat{V}_r=I, & A_r &= \hat{W}_r^* A_t \hat{V}_r, &B_r &= \hat{W}_r^* B_t, & C_r &= C_t \hat{V}_r.\nonumber
\end{align}
\section{Low-rank ADI-based Non-intrusive Balanced Truncation for Continuous-time Systems}\label{sec3}
In this section, we propose a non-intrusive implementation of BT using transfer function samples from the right-half of the $s$-plane, as opposed to the $j\omega$-axis, which is used for QuadBT. 

Projection-based low-rank methods for Lyapunov equations approximate the Lyapunov equations (\ref{lyap_P}) and (\ref{lyap_Q}) as follows:
\begin{align}
P &\approx \tilde{V} \hat{P} \tilde{V}^*& Q&\approx \tilde{W} \hat{Q}\tilde{W}^*.\nonumber
\end{align}
Any low-rank method for Lyapunov equations where \( \tilde{V} \) and \( \tilde{W} \) are interpolatory, and \( \hat{P} \) and \( \hat{Q} \) is computed non-intrusively can be effectively used to develop a non-intrusive low-rank BT algorithm. This is because, when \( \tilde{V} \) and \( \tilde{W} \) in \( P \approx \tilde{V} \hat{P} \tilde{V}^* \) and \( Q \approx \tilde{W} \hat{Q} \tilde{W}^* \), respectively, are interpolatory, the terms \( \tilde{W}^* E \tilde{V} \), \( \tilde{W}^* A \tilde{V} \), \( \tilde{W}^* B \), and \( C \tilde{V} \) can be computed non-intrusively within the Loewner framework using data. If \( \hat{P} = L_p L_p^* \) and \( \hat{Q} =  L_q L_q^* \) can also be computed non-intrusively, a non-intrusive formulation can be readily achieved.

The core idea behind interpolation-based methods and frequency-domain quadrature-based methods for approximating the Lyapunov equations (\ref{lyap_P}) and (\ref{lyap_Q}) is fundamentally similar. In numerical integration, the integrand is approximated by constructing its interpolant at specific nodes, which then serves as a surrogate for the original integrand in the integral. Instead of directly computing the integral of the original function, the integral of the interpolant is evaluated. Interpolatory projection-based methods implicitly follow the same approach. First, the interpolant of \( X(s) = (sE - A)^{-1}BB^T(s^*E^T - A^T)^{-1} \) is constructed as follows:
\[
\tilde{X}(s) = \tilde{V}(s\tilde{E} - \tilde{A})^{-1}\tilde{B}\tilde{B}^*(s^*\tilde{E}^* - \tilde{A}^*)^{-1}\tilde{V}^*,
\]
where \( X(\sigma_i) = \tilde{X}(\sigma_i) \) for \( i = 1, \dots, n_p \), and \(\sigma_i\) represent the chosen interpolation points. Subsequently, the method approximates \( P \) by implicitly computing the integral:
\[
P \approx \frac{1}{2\pi} \int_{-\infty}^{\infty} \tilde{X}(j\omega) \, d\omega.
\]
The key distinction lies in where \( X(s) \) is interpolated: in numerical integration, \( X(s) \) is interpolated along the \( j\omega \)-axis, whereas in interpolatory projection methods like the ADI method, the interpolation occurs in the right-half of the \( s \)-plane. However, the essence of interpolation is the same in both methods—replacing \( X(j\omega) \) with its approximation.

In \cite{wolf2016adi}, it is shown that the low-rank approximation of Lyapunov equations produced by the block version of PORK is identical to that produced by the ADI method \cite{benner2013efficient} when the mirror images of the interpolation points are used as ADI shifts. The block version of PORK enforces block interpolation instead of tangential interpolation. Over the past few decades, the ADI method has been highly successful in extending the applicability of BT to large-scale systems \cite{saak2012goal,benner2014self}. In the sequel, a non-intrusive implementation of the block version of PORK-based BT is formulated, which produces results identical to the ADI-based BT.

The controllability Gramian \( \hat{P} \) of the ROM produced by I-PORK is given by \( \hat{P} = Q_s^{-1} \). Similarly, the observability Gramian \( \hat{Q} \) of the ROM produced by O-PORK is given by \( \hat{Q} = P_s^{-1} \). These Gramians can be computed non-intrusively using only interpolation data. Furthermore, the projection matrices in I-PORK and O-PORK, respectively, are interpolatory. Thus, PORK qualifies for use in the non-intrusive implementation of low-rank BT.

In block interpolation, the projection matrices
\begin{align}
\tilde{V} &= \begin{bmatrix} (\sigma_1 E - A)^{-1} B & \cdots & (\sigma_{n_p} E - A)^{-1} B \end{bmatrix},\nonumber\\
\tilde{W} &= \begin{bmatrix} (\mu_1^* E^T - A^T)^{-1} C^T & \cdots & (\mu_{n_q}^* E^T - A^T)^{-1} C^T \end{bmatrix},\nonumber
\end{align}
solve the following Sylvester equations:
\begin{align}
A \tilde{V} - E \tilde{V} S_b + B L_b &= 0,\nonumber\\
A^T \tilde{W} - E^T \tilde{W} S_c^* + C^T L_c^T &= 0,\nonumber
\end{align}
where
\begin{align}
S_b &= \text{diag}(\sigma_1, \dots, \sigma_{n_p})\otimes I_m, &S_c &= \text{diag}(\mu_1, \dots, \mu_{n_q})\otimes I_p,\nonumber\\
L_b &= \begin{bmatrix} 1 & \cdots & 1 \end{bmatrix}\otimes I_m, & L_c^T &= \begin{bmatrix} 1 & \cdots &  1\end{bmatrix}\otimes I_p.\label{SbSc}
\end{align}
Assume that the pairs \((-S_b, L_b)\) and \((-S_c, L_c)\) are observable and controllable, respectively, and the Gramians \( Q_s \) and \( P_s \) solve the Lyapunov equations (\ref{pork_qs}) and (\ref{pork_ps}), respectively. The block version of PORK produces low-rank approximations of \( P \) and \( Q \) as \( P \approx \tilde{V} \hat{P} \tilde{V}^* \) and \( \tilde{W} \hat{Q} \tilde{W}^* \), where \( \hat{P} = Q_s^{-1} \) and \( \hat{Q} = P_s^{-1} \). These are the same approximations achieved using the ADI method with shifts \((-\sigma_1, \dots, -\sigma_{n_p})\) and \((-\mu_1, \dots, -\mu_{n_q})\), respectively.

Let us decompose \( \hat{P} = L_p L_p^* \) and \( \hat{Q} = L_q L_q^* \), and define \( \hat{Z}_p = V L_p \) and \( \hat{Z}_q = W L_q \). Thus, \( P \approx \hat{Z}_p \hat{Z}_p^* \) and \( Q \approx \hat{Z}_q \hat{Z}_q^* \). Low-rank BT can then be performed using these low-rank factors of the Gramians via the balancing square-root algorithm. As with Quad-BT, the expressions \( E_s=\tilde{W}^* E \tilde{V} \), \( A_s=\tilde{W}^* A \tilde{V} \), \( B_s=\tilde{W}^* B \), and \( C_s=C \tilde{V} \) can be computed non-intrusively within the Loewner framework, as follows:
\begin{align}
E_s&=\begin{bmatrix}-\frac{G(\sigma_1) - G(\mu_1)}{\sigma_1 - \mu_1}&\cdots&-\frac{G(\sigma_{n_p}) - G(\mu_1)}{\sigma_{n_p} - \mu_1}\\\vdots&\ddots&\vdots\\-\frac{G(\sigma_1) - G(\mu_{n_q})}{\sigma_1 - \mu_{n_q}}&\cdots&-\frac{G(\sigma_{n_p}) - G(\mu_{n_q})}{\sigma_{n_p} - \mu_{n_q}}\end{bmatrix},\nonumber\\
A_s&=\begin{bmatrix}-\frac{\sigma_1 G(\sigma_1) - \mu_1 G(\mu_1)}{\sigma_1 - \mu_1}&\cdots&-\frac{\sigma_{n_p} G(\sigma_{n_p}) - \mu_1 G(\mu_1)}{\sigma_{n_p} - \mu_1}\\\vdots&\ddots&\vdots\\-\frac{\sigma_1 G(\sigma_1) - \mu_{n_q} G(\mu_{n_q})}{\sigma_1 - \mu_{n_q}}&\cdots&-\frac{\sigma_{n_p} G(\sigma_{n_p}) - \mu_{n_q}G(\mu_{n_q})}{\sigma_{n_p} - \mu_{n_q}}\end{bmatrix},\nonumber\\
B_s&=\begin{bmatrix}G(\mu_1)\\\vdots\\G(\mu_{n_q})\end{bmatrix},\quad C_s=\begin{bmatrix}G(\sigma_1)&\cdots& G(\sigma_{n_p})\end{bmatrix}.\label{s_LF}
\end{align}
The low-rank factors \(\hat{Z}_p\) and \(\hat{Z}_q\) can then replace \(Z_p\) and \(Z_q\) in the balancing square-root algorithm as:
\begin{align}
L_q^TE_sL_p = \begin{bmatrix} \hat{U}_1 & \hat{U}_2 \end{bmatrix} \begin{bmatrix} \hat{S}_1 & 0 \\ 0 & \hat{S}_2 \end{bmatrix} \begin{bmatrix} \hat{V}_1^* \\ \hat{V}_2^* \end{bmatrix}.\label{proj_svd}
\end{align}
Further, let the projection matrices \( \hat{W}_r \) and \( \hat{V}_r \) be defined as follows:
\begin{align}
\hat{W}_r = L_q \hat{U}_1 \hat{S}_1^{-\frac{1}{2}} \quad \text{and} \quad \hat{V}_r = L_p \hat{V}_1 \hat{S}_1^{-\frac{1}{2}}.\label{proj_mat}
\end{align}
The ROM in low-rank ADI-based BT is constructed by reducing the Loewner quadruplet \((E_s, A_s, B_s, C_s)\) as follows:
\begin{align}
E_r &= \hat{W}_r^* E_s \hat{V}_r=I, & A_r &= \hat{W}_r^* A_s \hat{V}_r, &B_r &= \hat{W}_r^* B_s, & C_r &= C_s \hat{V}_r.\label{dist_Es}
\end{align} The pseudo-code for the non-intrusive ADI-based BT (NI-ADI-BT) is provided in Algorithm \ref{alg3}.
\begin{algorithm}
\caption{NI-ADI-BT}
\textbf{Input:} ADI shifts for approximating $P$: $(-\sigma_1,\cdots,-\sigma_{n_p})$; ADI shifts for approximating $Q$: $(-\mu_1,\cdots,-\mu_{n_q})$; Frequency-domain data: $\big(G(\sigma_1),\cdots,G(\sigma_{n_p}),G(\mu_1),\cdots,G(\mu_{n_q})\big)$ and $G^\prime(\sigma_i)$ for $\sigma_i=\mu_j$; Reduced order: $r$.

\textbf{Output:} ROM: $(E_r,A_r,B_r,C_r)$
\begin{algorithmic}[1]\label{alg3}
\STATE Compute the Loewner quadruplet $(E_s,A_s,B_s,C_s)$ from (\ref{s_LF}).\label{step1}
\STATE Set $S_b$, $S_c$, $L_b$, and $L_c$ as in (\ref{SbSc}).
\STATE Compute $Q_s$ and $P_s$ by solving the Lyapunov equations (\ref{pork_qs}) and (\ref{pork_ps}).
\STATE Decompose $Q_s^{-1}=L_pL_p^*$ and $P_s^{-1}=L_qL_q^*$.\label{step4}
\STATE Compute the projection matrices $\hat{V}_r$ and $\hat{W}_r$ from (\ref{proj_svd}) and (\ref{proj_mat}).
\STATE Compute the ROM from (\ref{dist_Es}).
\end{algorithmic}
\end{algorithm}

Furthermore, the low-rank Gramians produced by PORK monotonically approach the original Gramians as the number of interpolation points increases. Note that PORK satisfies the following:
\begin{align}
||G(s)-G_r(s)||_{\mathcal{H}_2}^2&=\text{trace}(CPC^T)-\text{trace}(CVQ_s^{-1}V^*C^T)\nonumber\\
&=\text{trace}\big(C(P-VQ_s^{-1}V^*)C^T\big),\label{res_V}\\
||G(s)-G_r(s)||_{\mathcal{H}_2}^2&=\text{trace}(B^TQB)-\text{trace}(B^TWP_s^{-1}W^*B)\nonumber\\
&=\text{trace}\big(B^T(Q-WP_s^{-1}W^*)B\big),\label{res_W}
\end{align}
The only variable part in (\ref{res_V}) is \( \text{trace}(C V Q_s^{-1} V^* C^T) =\text{trace}(C_sQ_s^{-1} C_s^*)\), which grows monotonically as the number of interpolation points increases. Similarly, the only variable part in (\ref{res_W}) is \( \text{trace}(B^T W P_s^{-1} W^* B) =\text{trace}(B_s^* P_s^{-1} B_s)\), which also grows monotonically as the number of interpolation points increases. Both these terms can be computed non-intrusively, allowing us to quantify the improvement in the accuracy of the Gramians by monitoring their growth.
\subsection{Avoiding the Computation of $Q_s$ and $P_s$}
In practice, as the number of interpolation points increases, \( Q_s \) and \( P_s \) start losing numerical rank, making the computation of inverses \( Q_s^{-1} \) and \( P_s^{-1} \) highly ill-conditioned. This situation can be avoided if we make the choice of ADI shifts \(-\sigma_i\) and \(-\mu_i\) a bit restrictive. Let us represent \(-\sigma_i\) and \(-\mu_i\) as follows:
\begin{align}
-\sigma_i=-\Bigg(\frac{\zeta_{i,\sigma}|\omega_{i,\sigma}|}{\sqrt{1-\zeta_{i,\sigma}^2}}\Bigg)+j\omega_{i,\sigma}\quad\textnormal{and}\quad -\mu_i=-\Bigg(\frac{\zeta_{i,\mu}|\omega_{i,\mu}|}{\sqrt{1-\zeta_{i,\mu}^2}}\Bigg)+j\omega_{i,\mu},\label{adi_shifts_cont}
\end{align}where $0<\zeta_{i,\sigma}\leq 1$ and $0<\zeta_{i,\mu}\leq 1$ are the damping coefficients of the ADI shifts. When \(\zeta_{i,\sigma}\ll 1\) and \(\zeta_{i,\mu}\ll 1\), the Gramians \(Q_s\) and \(P_s\) approximate the following:  
\begin{align}
Q_s&\approx \textnormal{diag}\Big(\frac{1}{2\textnormal{Re}(\sigma_1)},\cdots,\frac{1}{2\textnormal{Re}(\sigma_{n_p})}\Big)\otimes I_m,\label{qs_formula}\\
P_s&\approx \textnormal{diag}\Big(\frac{1}{2\textnormal{Re}(\mu_1)},\cdots,\frac{1}{2\textnormal{Re}(\mu_{n_q})}\Big)\otimes I_p;\label{ps_formula}
\end{align}see \cite{gawronski2004dynamics} for a detailed analysis of Gramians for modal state-space models. Thus, \( L_p \) and \( L_q \) in this case approximate the following:
\begin{align}
L_p&\approx \textnormal{diag}\Big(\sqrt{2\textnormal{Re}(\sigma_1)},\cdots,\sqrt{2\textnormal{Re}(\sigma_{n_p})}\Big)\otimes I_m,\label{lp_formula}\\
L_q&\approx \textnormal{diag}\Big(\sqrt{2\textnormal{Re}(\mu_1)},\cdots,\sqrt{2\textnormal{Re}(\mu_{n_q})}\Big)\otimes I_p.\label{lq_formula}
\end{align} In short, if the ADI shifts are lightly damped, we can compute \( L_p \) and \( L_q \) directly from the real parts of the shifts instead of computing \( Q_s \), \( P_s \), \( Q_s^{-1} \), and \( P_s^{-1} \).  
\subsection{Extension to Tangential-ADI Method}
In \cite{wolf2016adi}, a tangential version of the ADI method is proposed, based on the tangential PORK. To implement the tangential version of ADI-based BT, Algorithm \ref{alg3} can be slightly modified as follows: The tangential directions \( b_i \) and \( c_i \) should be normalized as \( \frac{b_i}{\|b_i\|_2} \) and \( \frac{c_i}{\|c_i\|_2} \), respectively, to avoid computing \( Q_s \) and \( P_s \). Note that this normalization does not affect the interpolation properties of the Loewner quadruplet but ensures that \( b_i^*b_i = c_i c_i^* = 1 \). In Step \ref{step1} of Algorithm \ref{alg3}, the Loewner quadruplet can be computed from (\ref{LF}). Assuming the ADI shifts are lightly damped, \( L_p \) and \( L_q \) in Step \ref{step4} can be computed as follows:
\begin{align}
L_p&\approx \textnormal{diag}\Big(\sqrt{2\textnormal{Re}(\sigma_1)},\cdots,\sqrt{2\textnormal{Re}(\sigma_{n_p})}\Big),\nonumber\\
L_q&\approx \textnormal{diag}\Big(\sqrt{2\textnormal{Re}(\mu_1)},\cdots,\sqrt{2\textnormal{Re}(\mu_{n_q})}\Big).\nonumber
\end{align}These changes extend Algorithm \ref{alg3} to the tangential version of ADI-based BT.
\section{Non-intrusive Implementations of IRKA for Continuous-time Systems}\label{sec4}
IRKA is highly effective for constructing $\mathcal{H}_2$-optimal ROMs through iterative refinement of interpolation data. However, its non-intrusive implementation poses a significant practical challenge. Each IRKA iteration updates the interpolation points, necessitating new estimations of \( G(\sigma_i)b_i \), \( c_iG(\sigma_i) \), and \( c_iG^{\prime}(\sigma_i)b_i \). As a consequence, the algorithm must be paused to compute new estimations, making it unsuitable for practical applications. In this section, three non-intrusive implementations of IRKA are proposed, which rely on existing available data instead of requiring new estimations of \( G(\sigma_i)b_i \), \( c_iG(\sigma_i) \), and \( c_iG^{\prime}(\sigma_i)b_i \) each time IRKA updates the interpolation data triplet \((\sigma_i, b_i, c_i)\).
\subsection{Using Available Frequency Response Data}
In industries such as aerospace, defense, and automotive, frequency-domain data is collected to construct the Fourier transform \( G(j\omega) \) by exciting systems at various frequencies \( \omega \) rad/sec. This data plays a critical role in numerous analysis and design tasks, including system identification, control design, resonance frequency calculation, and vibration analysis, among others \citep{lennart1999system, ozbay2018frequency, pintelon2008frequency, gillberg2006frequency, morelli2020practical}. In this subsection, we demonstrate that this existing data is sufficient for non-intrusive implementation of IRKA.

When the interpolation points \( \sigma_i \) and \( \mu_i \) have positive real parts, \( V \) and \( W \) in (\ref{sylv_V}) and (\ref{sylv_W}), respectively, can be computed using the integral expressions:
\begin{align}
V &= \frac{1}{2\pi} \int_{-\infty}^{\infty} (j\nu E - A)^{-1} B L_b (-j\nu I + S_b)^{-1} d\nu,\label{int_V}\\
W^* &= \frac{1}{2\pi} \int_{-\infty}^{\infty} (-j\nu I + S_c)^{-1} L_c C (j\nu E - A)^{-1} d\nu,\label{int_W}
\end{align}cf. \citep{sorensen2002sylvester}. These integrals can be approximated using numerical integration as follows:
\begin{align}
V &\approx \frac{1}{2\pi} \sum_{i=1}^{n_p} w_{v,i} (j\omega_i E - A)^{-1} B L_b (-j\omega_i I + S_b)^{-1},\label{sum_v}\\ 
W^* &\approx \frac{1}{2\pi} \sum_{i=1}^{n_q} w_{w,i} (-j\nu_i I + S_c)^{-1} L_c C (j\nu_i E - A)^{-1}, \label{sum_w}
\end{align}
where \( \omega_i \) and \( \nu_i \) are nodes, and \( w_{v,i} \) and \( w_{w,i} \) are their respective weights.

Let us define the projection matrices\( \hat{V}_r \) and \( \hat{W}_r \) as follows:
\begin{align}
\hat{V}_r &= \frac{1}{2\pi} \begin{bmatrix} w_{v,1} L_b (-j\omega_1 I + S_b)^{-1} \\ \vdots \\ w_{v,n_p} L_b (-j\omega_{n_p} I + S_b)^{-1} \end{bmatrix},\label{V_r}\\
\hat{W}_r^* &= \frac{1}{2\pi} \begin{bmatrix} (-j\nu_1 I + S_c)^{-1} L_c w_{w,1} & \cdots & (-j\nu_{n_q} I + S_c)^{-1} L_c w_{w,n_q} \end{bmatrix}.\label{W_r}
\end{align}
It is evident that the summations (\ref{sum_v}) and (\ref{sum_w}) can be represented as \( \tilde{V} \hat{V}_r \) and \( \hat{W}_r^* \tilde{W}^* \), respectively, where $\tilde{V}$ and $\tilde{W}$ are as in (\ref{V_tild}) and (\ref{W_tild}), respectively. Thus, \( V \approx \tilde{V} \hat{V}_r \) and \( W \approx \tilde{W} \hat{W}_r \). Let us assume, for a moment, that this approximation is exact. In this case, the ROM satisfying the interpolation condition (\ref{int_cond_1}) can be obtained by reducing the Loewner quadruplet $(E_w,A_w,B_w,C_w)$ as follows:
\begin{align}
E_r &= \hat{W}_r^*E_w\hat{V}_r, & A_r &= \hat{W}_r^*A_w\hat{V}_r, &B_r &= \hat{W}_r^*B_w, & C_r &= C_w\hat{V}_r. \label{interim_ROM}
\end{align}
When $\sigma_j = \mu_i$, this ROM also satisfies the Hermite interpolation condition (\ref{int_cond_2}). Since $\hat{V}_r$ and $\hat{W}_r$ depend solely on the quadrature weights $w_{v,i}$ and $w_{w,i}$, the interpolation points $\sigma_j$ and $\mu_i$, and the tangential directions $b_j$ and $c_i$, the ROM $(E_r,A_r,B_r,C_r)$ can be computed non-intrusively.

It is now evident that IRKA can be implemented using frequency response data \( G(j\omega_i) \) and \( G(j\nu_i) \), eliminating the need for repeated estimations of \( G(\sigma_i) \) and \( G^\prime(\sigma_i) \) whenever IRKA updates $\sigma_i$. The pseudo-code for our proposed algorithm, called ``frequency-domain quadrature-based IRKA (FD-Quad-IRKA)'', is provided in Algorithm \ref{alg1}.
\begin{algorithm}
\caption{FD-Quad-IRKA}
\textbf{Inputs:}  Nodes: $(\omega_1,\cdots,\omega_{n_p})$, $(\nu_1,\cdots,\nu_{n_q})$; Frequency-domain data: $\big(G(j\omega_1),\cdots,G(j\omega_{n_p})\big)$, $\big(G(j\nu_1),\cdots,G(j\nu_{n_q})\big)$; $G^\prime(j\nu_i)$ for $\omega_i=\nu_j$; Quadrature weights: $(w_{v,1},\cdots,w_{v,n_p})$, $(w_{w,1},\cdots,w_{v,n_q})$; Interpolation data: $(\sigma_1,\cdots,\sigma_r)$, $(b_1,\cdots,b_r)$, $(c_1,\cdots,c_r)$; Tolerance: tol.

\textbf{Outputs:} ROM: $(E_r,A_r,B_r,C_r)$
\begin{algorithmic}[1]\label{alg1}
\STATE Compute the Loewner quadruplet $(E_w, A_w, B_w, C_w)$ from (\ref{jw_LF}). 
\STATE \textbf{while}\Big(relative change in $\lambda_i$ > tol\Big)
\STATE Set $S_b$, $L_b$, $S_c$, and $L_c$ as in (\ref{SbLbScLc}).
\STATE Set the projection matrices $\hat{V}_r$ and $\hat{W}_r$ as in (\ref{V_r}) and (\ref{W_r}).
\STATE Compute $(E_r,A_r,B_r,C_r)$ from (\ref{interim_ROM}).
\STATE Compute the eigenvalue decomposition: $E_r^{-1}A_r=T_r\Lambda T_r^{-1}$ where $\Lambda=diag(\lambda_1,\cdots,\lambda_r)$.
\STATE Update the interpolation data: $(\sigma_1,\cdots,\sigma_r)=(-\lambda_1,\cdots,-\lambda_r)$; $[b_1\cdots b_r]=B_r^*E_r^{-*}T_r^{-*}$; $[c_1^*\cdots c_r^*]=C_rT_r$.
\STATE \textbf{end while}
\end{algorithmic}
\end{algorithm}

\textit{Range of Frequency Domain Sampling:}
Let us restrict the integral range of (\ref{int_V}) from $[-\infty,\infty]$ to $[-\nu,\nu]$ rad/sec. Then $V_\nu=V\Big|_{-\nu}^{\nu}$ solves the following Sylvester equation:
\[
AV_\nu - EV_\nu S_b + S_{\nu,a} B L_b + B L_b S_{\nu,s} = 0,
\]
where
\[
S_{\nu,a} = \frac{E}{2\pi} \int_{-\nu}^{\nu} (j\nu E - A)^{-1} d\nu, \quad S_{\nu,s} = \frac{1}{2\pi} \int_{-\nu}^{\nu} (j\nu I + S_b)^{-1} d\nu,
\]
as described in \cite{zulfiqar2020frequency}. Theoretically, $S_{\nu,a} \rightarrow \frac{1}{2}I$ and $S_{\nu,s} \rightarrow \frac{1}{2}I$ as $\nu \rightarrow \infty$. In practice, $S_{\nu,a}$ reduces to $\frac{1}{2}I$ outside the bandwidth of $G(s)$. Similarly, $S_{\nu,s}$ begins to approach $\frac{1}{2}I$ once $\nu$ exceeds the largest imaginary part of the eigenvalues of $S_b$. As a result, $V_\nu$ becomes numerically equivalent to $V$ beyond a finite frequency range. Therefore, in practice, the nodes of the numerical quadrature can be confined to a finite frequency range, especially when the bandwidth of the system $G(s)$ is known.

Alternatively, the integration limits of the numerical quadrature rule can be mapped to $[-\infty,\infty]$. For instance, the integration limits $[-1, 1]$ in the Gauss-Legendre quadrature rule can be mapped to $[-\infty,\infty]$ using the following transformation:
\[
y = \tan\left(\frac{\pi}{2} x\right), \quad \frac{dy}{dx} = \frac{\pi}{2} \sec^2\left(\frac{\pi}{2} x\right).
\]
The quadrature weights can then be adjusted as $w_y = w_x \frac{\pi}{2} \sec^2\left(\frac{\pi}{2} x\right)$.

\textit{Avoiding the Samples of $G^{\prime}(s)$:} To implement IRKA, transfer function samples and their derivatives at \( r \) interpolation points are required. These can be computed by numerically integrating a single integral, though with some loss of accuracy compared to Algorithm \ref{alg1}, since single-sided moment matching is less accurate than double-sided moment matching. Let us define \( S_b \) and \( L_b \) as follows:
\begin{align}
S_b &= \text{blkdiag}\left( \begin{bmatrix} \sigma_1 & 1 \\ 0 & \sigma_1 \end{bmatrix}, \cdots, \begin{bmatrix} \sigma_r & 1 \\ 0 & \sigma_r \end{bmatrix} \right)\otimes I_m,\nonumber\\
L_b &= \begin{bmatrix} 1 & 0 & \cdots & 1 & 0 \end{bmatrix}\otimes I_m. \label{S_L}
\end{align}
By solving the Sylvester equation (\ref{sylv_V}), we obtain:
\[
V = \Big[ (\sigma_1 E - A)^{-1} B \quad -(\sigma_1 E - A)^{-1} E (\sigma_1 E - A)^{-1} B \quad \cdots\]
\[(\sigma_r E - A)^{-1} B \quad -(\sigma_r E - A)^{-1} E (\sigma_r E - A)^{-1} B\Big],
\]
as described in \cite{wolf2014h}. Consequently,
\[
C V = \begin{bmatrix} G(\sigma_1) & G^\prime(\sigma_1) & \cdots & G(\sigma_r) & G^\prime(\sigma_r) \end{bmatrix}.
\]
Thus, by setting \( S_b \) and \( L_b \) as in (\ref{S_L}), we can compute:
\begin{align}
\begin{bmatrix} G(\sigma_1) & G^\prime(\sigma_1) & \cdots & G(\sigma_r) & G^\prime(\sigma_r) \end{bmatrix} \approx \begin{bmatrix} G(j\omega_1) & \cdots & G(j\omega_{n_p}) \end{bmatrix} \hat{V}_r.\nonumber
\end{align}
\subsection{Using Available Impulse Response Data}
In many applications, obtaining frequency-domain measurements is impractical. Instead, impulse response data is frequently utilized for various analysis and design tasks. When direct impulse response measurements are not feasible, a step input can be applied, and the impulse response can be obtained through differentiation. While a detailed review of these methods falls outside the scope of this paper, we refer readers to \citep{stan2002comparison, finno1998impulse, holters2009impulse, foster1986impulse, borish1983efficient} for further insights. In this subsection, we demonstrate that existing impulse response data is sufficient for non-intrusive implementation of IRKA.

If the interpolation points \(\sigma_i\) and \(\mu_i\) have positive real parts, \(V\) and \(W\) can be computed using the following integral expressions:
\begin{align}
V &= \int_{0}^{\infty} e^{E^{-1}A\tau} E^{-1} B L_b e^{-S_b\tau} d\tau,\label{int_Vt}\\
W^*& = \int_{0}^{\infty} e^{-S_c\tau} L_c C E^{-1} e^{A E^{-1}\tau} d\tau,\label{int_Wt}
\end{align}
where \(S_b = \text{diag}(\sigma_1, \dots, \sigma_r)\), \(S_c = \text{diag}(\mu_1, \dots, \mu_r)\), \(L_b = \begin{bmatrix} b_1, \dots, b_r \end{bmatrix}\), and \(L_c^* = \begin{bmatrix} c_1^*, \dots, c_r^* \end{bmatrix}\).

These integrals can be approximated using numerical integration as follows:
\begin{align}
V &\approx \sum_{i=1}^{n_p} w_{v,i} e^{E^{-1}A t_i} E^{-1} B L_b e^{-S_b t_i},\label{sum_vt}\\
W^* &\approx \sum_{i=1}^{n_q} w_{w,i} e^{-S_c \tau_i} L_c C E^{-1} e^{A E^{-1} \tau_i},\label{sum_wt}
\end{align}
where \(t_i\) and \(\tau_i\) are quadrature nodes, and \(w_{v,i}\) and \(w_{w,i}\) are their respective weights.

Let us define the projection matrices \(\hat{V}_r\) and \(\hat{W}_r\) as follows:
\begin{align}
\hat{V}_r &= \begin{bmatrix} w_{v,1} L_b e^{-S_b t_1} \\ \vdots \\ w_{v,n_p} L_b e^{-S_b t_{n_p}} \end{bmatrix},\label{V_r_t}\\
\hat{W}_r^* &= \begin{bmatrix} e^{-S_c \tau_1} L_c w_{w,1} & \cdots & e^{-S_c \tau_{n_q}} L_c w_{w,n_q} \end{bmatrix}.\label{W_r_t}
\end{align}
It is evident that the summations (\ref{sum_vt}) and (\ref{sum_wt}) can be represented as \(\tilde{V} \hat{V}_r\) and \(\hat{W}_r^* \tilde{W}^*\), respectively, where $\tilde{V}$ and $\tilde{W}$ are as in (\ref{Vt_tild}) and (\ref{Wt_tild}), respectively. Thus, \(V \approx \tilde{V} \hat{V}_r\) and \(W \approx \tilde{W} \hat{W}_r\). Again, let us assume, for a moment, that this approximation is exact. In this case, the ROM satisfying the interpolation condition (\ref{int_cond_1}) can be obtained by reducing the impulse data quadruplet $(E_t,A_t,B_t,C_t)$ as follows:
\begin{align}
E_r&=\hat{W}_r^*E_t\hat{V}_r,& A_r&=\hat{W}_r^*A_t\hat{V}_r, & B_r&=\hat{W}_r^*B_t,& C_r&=C_t\hat{V}_r.\label{time_interimROM}
\end{align}
 When $\sigma_j = \mu_i$, this ROM also satisfies the Hermite interpolation condition (\ref{int_cond_2}). Since $\hat{V}_r$ and $\hat{W}_r$ depend solely on the quadrature weights $w_{v,i}$ and $w_{w,i}$, the interpolation points $\sigma_j$ and $\mu_i$, and the tangential directions $b_j$ and $c_i$, the ROM $(E_r,A_r,B_r,C_r)$ can be computed non-intrusively.

It is now evident that IRKA can be implemented using impulse response data, eliminating the need for repeated estimations of \( G(\sigma_i) \) and \( G^\prime(\sigma_i) \) whenever IRKA updates $\sigma_i$. The pseudo-code for our proposed algorithm, called ``time-domain quadrature-based IRKA (TD-Quad-IRKA)'', is provided in Algorithm \ref{alg2}.
\begin{algorithm}
\caption{TD-Quad-IRKA}
\textbf{Input:} Nodes: $(t_1,\cdots,t_{n_p}), (\tau_1,\cdots,\tau_{n_q})$; Impulse response data: $\big(h(t_1),\cdots,h(t_{n_p})\big)$, $\big(h(\tau_1),\cdots,h(\tau_{n_q})\big)$, $\big(h^\prime(t_1),\cdots,h^\prime(t_{n_p})\big)$, $\big(h^\prime(\tau_1),\cdots,h^\prime(\tau_{n_q})\big)$;  Quadrature weights: $(w_{v,1},\cdots,w_{v,n_p}), (w_{w,1},\cdots,w_{w,n_q})$; Interpolation data: $(\sigma_1,\cdots,\sigma_r)$, $(b_1,\cdots,b_r)$, $(c_1,\cdots,c_r)$; Tolerance: tol.

\textbf{Output:} ROM: $(E_r,A_r,B_r,C_r)$
\begin{algorithmic}[1]\label{alg2}
\STATE Compute the impulse data quadruplet $(E_t, A_t, B_t, C_t)$ from (\ref{impulse_quad}). 
\STATE \textbf{while}\Big(relative change in $\lambda_i$ > tol\Big)
\STATE Set $S_b$, $L_b$, $S_c$, and $L_c$ as in (\ref{SbLbScLc}).
\STATE Set the projection matrices $\hat{V}_r$ and $\hat{W}_r$ as in (\ref{V_r_t}) and (\ref{W_r_t}).
\STATE Compute $(E_r,A_r,B_r,C_r)$ from (\ref{time_interimROM}).
\STATE Compute the eigenvalue decomposition: $E_r^{-1}A_r=T_r\Lambda T_r^{-1}$ where $\Lambda=diag(\lambda_1,\cdots,\lambda_r)$.
\STATE Update the interpolation data: $(\sigma_1,\cdots,\sigma_r)=(-\lambda_1,\cdots,-\lambda_r)$; $[b_1\cdots b_r]=B_r^*E_r^{-*}T_r^{-*}$; $[c_1^*\cdots c_r^*]=C_rT_r$.
\STATE \textbf{end while}
\end{algorithmic}
\end{algorithm}

\textit{Range of Impulse Response Sampling:}
Let us restrict the integral range of (\ref{int_Vt}) from \([0, \infty]\) to \([0, t_f]\) rad/sec. Then \(V_\tau = V \big|_{0}^{t_f}\) solves the following Sylvester equation:
\[
A V_\tau - E V_\tau S_b + B L_b - E e^{E^{-1} A t_f} E^{-1} B L_b e^{-S_b t_f} = 0,
\]
as described in \cite{zulfiqar2020time}. Theoretically, \(e^{E^{-1} A t_f} \rightarrow 0\) and \(e^{-S_b t_f} \rightarrow 0\) as \(t_f \rightarrow \infty\). In practice, \(e^{E^{-1} A t_f}\) and \(e^{-S_b t_f}\) rapidly approach zero for a finite \(t_f\), depending on how far the eigenvalues of \(E^{-1} A\) and \(-S_b\) are from the \(j\omega\)-axis. The farther the eigenvalues of \(E^{-1} A\) and \(-S_b\) are from the \(j\omega\)-axis, the faster the exponentials \(e^{E^{-1} A t_f}\) and \(e^{-S_b t_f}\) decay to zero. As a result, \(V_\tau\) becomes numerically equivalent to \(V\) beyond a finite time range. Therefore, in practice, the nodes of the numerical quadrature can be confined to a finite time range, especially when the poles of \(G(s)\) are located far from the \(j\omega\)-axis in the left half of the \(s\)-plane. Consequently, we can use a finite \(t_f\) in the numerical quadrature rule, and the integration limits can be mapped accordingly. For instance, the integration limits \([-1, 1]\) in the Gauss-Legendre quadrature rule can be mapped to \([0, t_f]\) using the following transformation:
\[
y = 0.5 t_f (x + 1), \quad \frac{dy}{dx} = 0.5 t_f.
\]
The quadrature weights can then be adjusted as \(w_y = 0.5 t_f w_x\).

\textit{Avoiding the Samples of $h^{\prime}(t)$:} Similar to the frequency domain, by setting \(S_b\) and \(L_b\) as in (\ref{S_L}), we can compute the following:
\begin{align}
\begin{bmatrix} G(\sigma_1) & G^\prime(\sigma_1) & \cdots & G(\sigma_r) & G^\prime(\sigma_r) \end{bmatrix}\approx \begin{bmatrix} h(t_1) & \cdots & h(t_{n_p}) \end{bmatrix} \hat{V}_r.\nonumber
\end{align}Again, this results in some loss of accuracy, as single-sided moment matching is less accurate than double-sided moment matching.
\subsection{Using Available Transfer Function Samples}
In this subsection, we illustrate how the block version of PORK can be utilized to develop a non-intrusive implementation of IRKA using available transfer function samples. Recall the following expressions:
\begin{align}
CV&=\frac{1}{2\pi}\int_{-\infty}^{\infty}G(j\nu)L_b(-j\nu I+S_b)d\nu,\label{Gsig}\\
W^*B&=\frac{1}{2\pi}\int_{-\infty}^{\infty}(j\nu I+S_c)^{-1}L_cG(j\nu)d\nu.\label{Gmu}
\end{align}
Similar to the ADI method, if we substitute \( G(s) \) in (\ref{Gsig}) with its interpolant generated by the block version of I-PORK at the available interpolation points \( (\alpha_1, \cdots, \alpha_{n_p}) \), and replace \( G(s) \) in (\ref{Gmu}) with its interpolant produced by the block version of O-PORK at the interpolation points \( (\beta_1, \cdots, \beta_{n_q}) \), we can achieve a non-intrusive implementation of IRKA. The interpolation points \( \alpha_i \) and \( \beta_i \) are all located in the right-half of the \( s \)-plane.

Let us define the projection matrices \(\tilde{V}\) and \(\tilde{W}\) as follows:
\begin{align}
\tilde{V}&=\begin{bmatrix}(\alpha_1E-A)^{-1}B&\cdots&(\alpha_{n_p}E-A)^{-1}B\end{bmatrix},\\
\tilde{W}&=\begin{bmatrix}(\beta_1^*E^T-A^T)^{-1}C^T&\cdots&(\beta_{n_q}^*E^T-A^T)^{-1}C^T\end{bmatrix}.
\end{align}
Additionally, let us define the following matrices:
\begin{align}
S_\alpha&=\text{diag}(\alpha_1,\cdots,\alpha_{n_p})\otimes I_m,& S_\beta&=\text{diag}(\beta_1,\cdots,\beta_{n_q})\otimes I_p,\nonumber\\
L_\alpha&=\begin{bmatrix}1&\cdots&1\end{bmatrix}\otimes I_m, & L_{\beta}^T&=\begin{bmatrix}1&\cdots&1\end{bmatrix}\otimes I_p.
\end{align}
Let \(Q_\alpha\) and \(P_\beta\) be the solutions to the following Lyapunov equations:
\begin{align}
-S_\alpha^*Q_\alpha-Q_\alpha S_\alpha +L_{\alpha}^TL_\alpha&=0,\\
-S_\beta P_\beta-P_\beta S_\beta^* +L_\beta L_\beta^T&=0.
\end{align}
In (\ref{Gsig}), \(G(s)\) can be replaced with the ROM produced by the block version of I-PORK:
\begin{align}
E_\alpha&=I,&A_\alpha&=-Q_\alpha^{-1}S_\alpha^*Q_\alpha,\nonumber\\
B_\alpha&=Q_\alpha^{-1}L_\alpha^T,& C_\alpha&=C\tilde{V}.
\end{align}
Similarly, in (\ref{Gmu}), \(G(s)\) can be replaced with the ROM produced by the block version of O-PORK:
\begin{align}
E_\beta&=I,&A_\beta&=-P_\beta S_\beta^*P_\beta^{-1},\nonumber\\
B_\beta&=\tilde{W}^*B,& C_\beta&=L_\beta^*P_\beta^{-1}.
\end{align}
Consequently, we obtain the following approximations:
\begin{align}
CV&\approx\frac{1}{2\pi}C\tilde{V}\Big(\int_{-\infty}^{\infty}(j\nu I-A_\alpha)^{-1}B_\alpha L_b(-j\nu I+S_b)d\nu\Big),\\
W^*B&\approx\Big(\frac{1}{2\pi}\int_{-\infty}^{\infty}(j\nu I+S_c)^{-1}L_cC_\beta(j\nu I-A_\beta)^{-1}d\nu\Big)\tilde{W}^*B.
\end{align}
Let the projection matrices \(\hat{V}_r\) and \(\hat{W}_r\) be defined as:
\begin{align}
\hat{V}_r=\begin{bmatrix}(\sigma_1 I-A_\alpha)^{-1}B_\alpha b_1&\cdots&(\sigma_r I-A_\alpha)^{-1}B_\alpha b_r\end{bmatrix},\label{Var}\\
\hat{W}_r=\begin{bmatrix}(\mu_1^*I-A_\beta^*)^{-1}C_\beta^*c_1^*&\cdots&(\mu_r^*I-A_\beta^*)^{-1}C_\beta^*c_r^*\end{bmatrix}.\label{Wbr}
\end{align}
It can then be observed that \(V \approx \tilde{V} \hat{V}_r\) and \(W \approx \tilde{W} \hat{W}_r\). Assuming this approximation is exact, the ROM satisfying the interpolation condition (\ref{int_cond_1}) can be obtained by reducing the Loewner quadruplet \((E_{\alpha,\beta}, A_{\alpha,\beta}, B_{\alpha,\beta}, C_{\alpha,\beta}) = (\tilde{W}^* E \tilde{V}, \tilde{W}^* A \tilde{V}, \tilde{W}^* B, C \tilde{V})\) as follows:
\begin{align}
E_r &= \hat{W}_r^*E_{\alpha,\beta}\hat{V}_r, & A_r &= \hat{W}_r^*A_{\alpha,\beta}\hat{V}_r, &B_r &= \hat{W}_r^*B_{\alpha,\beta}, & C_r &= C_{\alpha,\beta}\hat{V}_r, \label{a_b_interim_ROM}
\end{align}
where
\begin{align}
E_{\alpha,\beta}&=\begin{bmatrix}-\frac{G(\alpha_1) - G(\beta_1)}{\alpha_1 - \beta_1}&\cdots&-\frac{G(\alpha_{n_p}) - G(\beta_1)}{\alpha_{n_p} - \beta_1}\\\vdots&\ddots&\vdots\\-\frac{G(\alpha_1) - G(\beta_{n_q})}{\alpha_1 - \beta_{n_q}}&\cdots&-\frac{G(\alpha_{n_p}) - G(\beta_{n_q})}{\alpha_{n_p} - \beta_{n_q}}\end{bmatrix},\nonumber\\
A_{\alpha,\beta}&=\begin{bmatrix}-\frac{\alpha_1 G(\alpha_1) - \beta_1 G(\beta_1)}{\alpha_1 - \beta_1}&\cdots&-\frac{\alpha_{n_p} G(\alpha_{n_p}) - \beta_1 G(\beta_1)}{\alpha_{n_p} - \beta_1}\\\vdots&\ddots&\vdots\\-\frac{\alpha_1 G(\alpha_1) - \beta_{n_q} G(\beta_{n_q})}{\alpha_1 - \beta_{n_q}}&\cdots&-\frac{\alpha_{n_p} G(\alpha_{n_p}) - \beta_{n_q}G(\beta_{n_q})}{\alpha_{n_p} - \beta_{n_q}}\end{bmatrix},\nonumber\\
B_{\alpha,\beta}&=\begin{bmatrix}G(\beta_1)\\\vdots\\G(\beta_{n_q})\end{bmatrix},\quad C_{\alpha,\beta}=\begin{bmatrix}G(\alpha_1)&\cdots& G(\alpha_{n_p})\end{bmatrix}.\label{a_b_LF}
\end{align}
When \(\sigma_j = \mu_i\), this ROM also satisfies the Hermite interpolation condition (\ref{int_cond_2}). Since \(\hat{V}_r\) and \(\hat{W}_r\) depend only on the interpolation points \(\alpha_j\), \(\beta_i\), \(\sigma_j\), and \(\mu_i\), as well as the tangential directions \(b_j\) and \(c_i\), the ROM \((E_r, A_r, B_r, C_r)\) can be computed in a non-intrusive manner.

It is now evident that IRKA can be implemented using available transfer function samples \( G(\alpha_i) \) and \( G(\beta_i) \), eliminating the need for repeated estimations of \( G(\sigma_i) \) and \( G^\prime(\sigma_i) \) whenever IRKA updates $\sigma_i$. The pseudo-code for our proposed algorithm, called ``PORK-based IRKA (PORK-IRKA)'', is provided in Algorithm \ref{alg02}.
\begin{algorithm}
\caption{PORK-IRKA}
\textbf{Inputs:}  Sampling points: $(\alpha_1,\cdots,\alpha_{n_p})$, $(\beta_1,\cdots,\beta_{n_q})$; Transfer function samples: $\big(G(\alpha_1),\cdots,G(\alpha_{n_p})\big)$, $\big(G(\beta_1),\cdots,G(\beta_{n_q})\big)$; $G^\prime(\alpha_i)$ for $\alpha_i=\beta_j$; Interpolation data: $(\sigma_1,\cdots,\sigma_r)$, $(b_1,\cdots,b_r)$, $(c_1,\cdots,c_r)$; Tolerance: tol.

\textbf{Outputs:} ROM: $(E_r,A_r,B_r,C_r)$
\begin{algorithmic}[1]\label{alg02}
\STATE Compute the Loewner quadruplet $(E_{\alpha,\beta}, A_{\alpha,\beta}, B_{\alpha,\beta}, C_{\alpha,\beta})$ from (\ref{a_b_LF}). 
\STATE \textbf{while}\Big(relative change in $\lambda_i$ > tol\Big)
\STATE Set the projection matrices $\hat{V}_r$ and $\hat{W}_r$ as in (\ref{Var}) and (\ref{Wbr}).
\STATE Compute $(E_r,A_r,B_r,C_r)$ from (\ref{a_b_interim_ROM}).
\STATE Compute the eigenvalue decomposition: $E_r^{-1}A_r=T_r\Lambda T_r^{-1}$ where $\Lambda=diag(\lambda_1,\cdots,\lambda_r)$.
\STATE Update the interpolation data: $(\sigma_1,\cdots,\sigma_r)=(-\lambda_1,\cdots,-\lambda_r)$; $[b_1\cdots b_r]=B_r^*E_r^{-*}T_r^{-*}$; $[c_1^*\cdots c_r^*]=C_rT_r$.
\STATE \textbf{end while}
\end{algorithmic}
\end{algorithm}
\subsection{Tracking the Error $||G(s)-G_r(s)||_{\mathcal{H}_2}$}
Let \( G_r(s)^{(i-1)} \) and \( G_r(s)^{(i)} \) represent the interim ROMs in the \((i-1)^{th}\) and \(i^{th}\) iterations of IRKA, respectively. As noted in \cite{MPIMD11-11}, the error in the \((i-1)^{th}\) iteration can be computed after the \(i^{th}\) iteration as follows:
\begin{align}
||G(s)&-G_r(s)^{(i-1)}||_{\mathcal{H}_2}^2\nonumber\\
=&||G(s)||_{\mathcal{H}_2}^2+||G_r(s)^{(i-1)}||_{\mathcal{H}_2}^2-2\text{trace}\Big(C_r^{(i)}\big(C_r^{(i-1)}T_r^{(i-1)}\big)^*\Big).\nonumber
\end{align}Thus, with a delay of one iteration, the error \( ||G(s) - G_r(s)||_{\mathcal{H}_2} \) can be tracked if \( ||G_r(s)||_{\mathcal{H}_2}^2 \) is computed in every iteration. It is important to note that the original expression presented in \cite{MPIMD11-11} is intrusive, whereas the expression above is its non-intrusive equivalent. To summarize, the error in non-intrusive IRKA can also be monitored non-intrusively by tracking the following term:
\begin{align}
||G_r(s)^{(i-1)}||_{\mathcal{H}_2}^2-2\text{trace}\Big(C_r^{(i)}\big(C_r^{(i-1)}T_r^{(i-1)}\big)^*\Big).\nonumber
\end{align} However, it should be noted that the term \(-2 \, trace\Big(C_r^{(i)}\big(C_r^{(i-1)}T_r^{(i-1)}\big)^*\Big)\) is an approximation and not exact. Its accuracy depends on the precision of the approximations of the integrals (\ref{int_V}) and (\ref{int_W}) or (\ref{int_Vt}) and (\ref{int_Wt}).
\section{Non-intrusive Implementations of IRKA for Discrete-time Systems}\label{sec5}
Consider the following discrete-time system of order \( n \), denoted as \( G(z) \), and its ROM of order \( r \), denoted as \( G_r(z) \):
\begin{align}
G(z) &= C(zE - A)^{-1}B, \nonumber \\
G_r(z) &= C_r(zE_r - A_r)^{-1}B_r, \nonumber
\end{align}where \( z = e^{j\omega} \).

Assuming that \( G(z) \) and \( G_r(z) \) have simple poles, they can be expressed in the pole-residue form as follows:
\[
G(z) = \sum_{k=1}^{n} \frac{l_k r_k^*}{z - \lambda_k}, \quad G_r(z) = \sum_{k=1}^{r} \frac{\hat{l}_k \hat{r}_k^*}{z - \hat{\lambda}_k}.
\]
The necessary conditions for a local optimum of \( ||G(z) - G_r(z)||_{\mathcal{H}_2}^2 \) are given by:
\begin{align}
\hat{l}_i^* G^{\prime}\left(\frac{1}{\hat{\lambda}_i}\right) \hat{r}_i &= \hat{l}_i^* G_r^{\prime}\left(\frac{1}{\hat{\lambda}_i}\right) \hat{r}_i, \label{opd1} \\
\hat{l}_i^* G\left(\frac{1}{\hat{\lambda}_i}\right) &= \hat{l}_i^* G_r\left(\frac{1}{\hat{\lambda}_i}\right), \label{opd2} \\
G\left(\frac{1}{\hat{\lambda}_i}\right) \hat{r}_i &= G_r\left(\frac{1}{\hat{\lambda}_i}\right) \hat{r}_i, \label{opd3}
\end{align} for \( i = 1, 2, \dots, r \).

Similar to the continuous-time case, since the ROM \( G_r(z) \) is initially unknown, the discrete-time IRKA (DT-IRKA) \cite{bunse2010h2} uses fixed-point iterations starting from an arbitrary initial guess of the interpolation data to search for a local optimum. After each iteration, the interpolation data is updated as \( \sigma_i = \mu_i = \frac{1}{\hat{\lambda}_i} \), \( b_i = \hat{r}_i \), and \( c_i = \hat{l}_i^* \) until convergence is achieved. Upon convergence, a local optimum of \( ||G(z) - G_r(z)||_{\mathcal{H}_2}^2 \) is achieved. 

However, since DT-IRKA updates the interpolation points during the process, it requires estimating the transfer function samples at these updated points. This necessitates halting DT-IRKA and conducting new experiments to estimate new samples, which is often impractical.
\subsection{Using Available Frequency Response Data}
Let us define the following matrices:
\begin{align}
S_b&=\text{diag}(\sigma_1,\cdots,\sigma_r),\quad S_c=\text{diag}(\mu_1,\cdots,\mu_r),\nonumber\\
L_b&=\begin{bmatrix}b_1&\cdots&b_r\end{bmatrix},\quad L_c^*=\begin{bmatrix}c_1^*&\cdots&c_r^*\end{bmatrix},\nonumber\\
\bar{S}_b&=S_b^{-1},\quad \bar{L}_b=L_bS_b^{-1},\quad\bar{S}_c=S_c^{-1},\quad \bar{L}_c=S_c^{-1}L_c.\label{dt_SbLbScLC}
\end{align}
By post-multiplying equations (\ref{sylv_V}) and (\ref{sylv_W}) with \( S_b^{-1} \) and \( S_c^{-*} \), respectively, it can be observed that \( V \) and \( W \) in (\ref{Kry_V}) and (\ref{Kry_W}) satisfy the following Stein equations:  
\begin{align}
AV\bar{S}_b-EV+B\bar{L}_b&=0,\label{dsylv_V}\\
A^TW\bar{S}_c^*-E^TW+C^T\bar{L}_c^*&-0.\label{dsylv_W}
\end{align}When the eigenvalues of \( A \), \( \bar{S}_b \), and \( \bar{S}_c \) lie within the unit circle, \( V \) and \( W \) can be expressed using the following integral representations:  
\begin{align}
V&=\frac{1}{2\pi}\int_{-\pi}^{\pi}(e^{j\nu} E-A)^{-1}B\bar{L}_b(e^{-j\nu} I-\bar{S}_b)^{-1}d\nu,\label{dt_int_V}\\
W^*&=\frac{1}{2\pi}\int_{-\pi}^{\pi}(e^{-j\nu} I-\bar{S}_c)^{-1}\bar{L}_cC(e^{j\nu} E-A)^{-1}d\nu,\label{dt_int_W}
\end{align}cf. \cite{duan2015generalized}. These integrals can be approximated numerically as follows:  
\begin{align}
V&\approx\frac{1}{2\pi}\sum_{i=1}^{n_p}w_{v,i}(e^{j\xi_i} E-A)^{-1}B\bar{L}_b(e^{-j\xi_i} I-\bar{S}_b)^{-1},\label{sum_dv}\\
W^*&\approx\frac{1}{2\pi}\sum_{i=1}^{n_q}w_{w,i}(e^{-j\zeta_i} I-\bar{S}_c)^{-1}\bar{L}_cC(e^{j\zeta_i} E-A)^{-1},\label{sum_dw}
\end{align}where \( \xi_i \) and \( \zeta_i \) are the nodes, and \( w_{v,i} \) and \( w_{w,i} \) are their corresponding weights.
Next, define the following matrices:  
\begin{align}
\tilde{V}&=\begin{bmatrix}(e^{j\xi_1} E-A)^{-1}B&\cdots&(e^{j\xi_{n_p}} E-A)^{-1}B\end{bmatrix},\\
\hat{V}_r&=\frac{1}{2\pi}\begin{bmatrix}w_{v,1}\bar{L}_b(e^{-j\xi_1} I-\bar{S}_b)^{-1}\\\vdots\\w_{v,n_p}\bar{L}_b(e^{-j\xi_{n_p}} I-\bar{S}_b)^{-1}\end{bmatrix},\label{dt_Vr}\\
\tilde{W}^*&=\begin{bmatrix}C(e^{j\zeta_1} E-A)^{-1}\\\vdots\\C(e^{j\zeta_{n_q}} E-A)^{-1}\end{bmatrix},\\
\hat{W}_r^*&=\frac{1}{2\pi}\begin{bmatrix}(e^{-j\zeta_1} I-\bar{S}_c)^{-1}\bar{L}_cw_{w,1}&\cdots&(e^{-j\zeta_{n_q}} I-\bar{S}_c)^{-1}\bar{L}_cw_{w,n_q}\end{bmatrix}.\label{dt_Wr}
\end{align} From these definitions, it is clear that the summations (\ref{sum_dv}) and (\ref{sum_dw}) can be represented as \( \tilde{V} \hat{V}_r \) and \( \hat{W}_r^* \tilde{W}^* \), respectively. Thus, \( V \approx \tilde{V} \hat{V}_r \) and \( W \approx \tilde{W} \hat{W}_r \). Let us assume, for a moment, that this approximation is exact. In this case, the ROM satisfying the interpolation condition (\ref{int_cond_1}) can be obtained by reducing the Loewner quadruplet $(E_{jw},A_{jw},B_{jw},C_{jw})=(\tilde{W}^*E\tilde{V},\tilde{W}^*A\tilde{V},\tilde{W}^*B,C\tilde{V})$ as follows:
\begin{align}
E_r &= \hat{W}_r^*E_{jw}\hat{V}_r, & A_r &= \hat{W}_r^*A_{jw}\hat{V}_r, &B_r &= \hat{W}_r^*B_{jw}, & C_r &= C_{jw}\hat{V}_r, \label{jw_interim_ROM}
\end{align}
where 
\begin{align}
E_{jw}&=\begin{bmatrix}-\frac{G(e^{j\xi_1}) - G(e^{j\zeta_1})}{e^{j\xi_1} - e^{j\zeta_1}}&\cdots&-\frac{G(e^{j\xi_{n_p}}) - G(e^{j\zeta_1})}{e^{j\xi_{n_p}} - e^{j\zeta_1}}\\\vdots&\ddots&\vdots\\-\frac{G(e^{j\xi_1}) - G(e^{j\zeta_{n_q}})}{e^{j\xi_1} - e^{j\zeta_{n_q}}}&\cdots&-\frac{G(e^{j\xi_{n_p}}) - G(e^{j\zeta_{n_q}})}{e^{j\xi_{n_p}} -e^{j\zeta_{n_q}}}\end{bmatrix},\nonumber\\
A_{jw}&=\begin{bmatrix}-\frac{e^{j\xi_1} G(e^{j\xi_1}) - e^{j\zeta_1} G(e^{j\zeta_1})}{e^{j\xi_1} - e^{j\zeta_1}}&\cdots&-\frac{e^{j\xi_{n_p}} G(e^{j\xi_{n_p}}) - e^{j\zeta_1} G(e^{j\zeta_1})}{e^{j\xi_{n_p}} - e^{j\zeta_1}}\\\vdots&\ddots&\vdots\\-\frac{e^{j\xi_1} G(e^{j\xi_1}) - e^{j\zeta_{n_q}} G(e^{j\zeta_{n_q}})}{e^{j\xi_1} - e^{j\zeta_{n_q}}}&\cdots&-\frac{e^{j\xi_{n_p}} G(e^{j\xi_{n_p}}) - e^{j\zeta_{n_q}} G(e^{j\zeta_{n_q}})}{e^{j\xi_{n_p}} - e^{j\zeta_{n_q}}}\end{bmatrix},\nonumber\\
B_{jw}&=\begin{bmatrix}G(e^{j\zeta_1})\\\vdots\\G(e^{j\zeta_{n_q}})\end{bmatrix},\quad C_{jw}=\begin{bmatrix}G(e^{j\xi_1})&\cdots& G(e^{j\xi_{n_p}})\end{bmatrix}.\label{ejw_LF}
\end{align} Note that this is the same Loewner quadruplet $(E_{jw}, A_{jw}, B_{jw}, C_{jw})$ that the frequency-domain discrete-time QuadBT reduces to obtain a truncated balanced model, as discussed in \cite{goseaQuad}. When $\sigma_j = \mu_i$, this ROM also satisfies the Hermite interpolation condition (\ref{int_cond_2}). Since $\hat{V}_r$ and $\hat{W}_r$ depend solely on the quadrature weights $w_{v,i}$ and $w_{w,i}$, the interpolation points $\sigma_j$ and $\mu_i$, and the tangential directions $b_j$ and $c_i$, the ROM $(E_r,A_r,B_r,C_r)$ can be computed non-intrusively.

It is now clear that DT-IRKA can be implemented using frequency-domain data \( G(e^{j\xi_i}) \), eliminating the need for repeated estimations of \( G(\sigma_i) \) and \( G^\prime(\sigma_i) \) whenever DT-IRKA updates $\sigma_i$. The pseudo-code for the frequency-domain quadrature-based DT-IRKA (FD-Quad-DTIRKA) is outlined in Algorithm \ref{alg4}.
\begin{algorithm}
\caption{FD-Quad-DTIRKA}
\textbf{Input:} Nodes: $(\xi_1,\cdots,\xi_{n_p})$, $(\zeta_1,\cdots,\zeta_{n_q})$; Frequency-domain data: $\big(G(e^{j\xi_1}),\cdots,G(e^{j\xi_{n_v}})\big)$, $\big(G(e^{j\zeta_1}),\cdots,G(e^{j\zeta_{n_q}})\big)$, $G^{\prime}(e^{j\xi_i})$ for $\xi_i=\zeta_j$; Quadrature weights: $(w_{v,1},\cdots,w_{v,n_p})$, $(w_{w,1},\cdots,w_{w,n_q})$; Interpolation data: $(\sigma_1,\cdots,\sigma_r)$, $(b_1,\cdots,b_r)$, $(c_1,\cdots,c_r)$; Tolerance: tol.

\textbf{Output:} ROM: $(E_r,A_r,B_r,C_r)$
\begin{algorithmic}[1]\label{alg4}
\STATE Compute the Loewner quadruplet $(E_{jw},A_{jw},B_{jw},C_{jw})$ from (\ref{ejw_LF}).
\STATE \textbf{while}\Big(relative change in $\lambda_i$ > tol\Big)
\STATE Set $\bar{S}_b$, $\bar{L}_b$, $\bar{S}_c$, and $\bar{L}_c$ as in (\ref{dt_SbLbScLC}).
\STATE Compute the projection matrices $\hat{V}_r$ and $\hat{W}_r$ from (\ref{dt_Vr}) and (\ref{dt_Wr}).
\STATE Compute $(E_r,A_r,B_r,C_r)$ from (\ref{jw_interim_ROM}).
\STATE Compute the eigenvalue decomposition: $E_r^{-1}A_r=T_r\Lambda T_r^{-1}$ where $\Lambda=diag(\lambda_1,\cdots,\lambda_r)$.
\STATE Update the interpolation data: $(\sigma_1,\cdots,\sigma_r)=(\frac{1}{\lambda_1},\cdots,\frac{1}{\lambda_r})$; $[b_1\cdots b_r]=B_r^*E_r^{-*}T_r^{-*}$; $[c_1^*\cdots c_r^*]=C_rT_r$.
\STATE \textbf{end while}
\end{algorithmic}
\end{algorithm}
\subsection{Using Available Impulse Response Data}
When the eigenvalues of \( A \) and \( \bar{S}_b \) lie within the unit circle, the projection matrices \( V \) and \( W \) in the Stein equations (\ref{dsylv_V}) and (\ref{dsylv_W}) can be expressed as the following infinite sums:
\begin{align}
V&=\sum_{i=0}^{\infty}(E^{-1}A)^iE^{-1}B\bar{L}_b\bar{S}_b^i,\label{dt_sum_V}\\
W^*&=\sum_{i=0}^{\infty}(\bar{S}_c)^i\bar{L}_cCE^{-1}(AE^{-1})^i.\label{dt_sum_W}
\end{align}Since the eigenvalues of \( A \) and \( \bar{S}_b \) are inside the unit circle, the terms \( A^i \) and \( \bar{S}_b^i \) decay as \( i \) increases. Consequently, after a finite number of terms, the summands \( (E^{-1}A)^i E^{-1} B \bar{L}_b \bar{S}_b^i \) and \( (\bar{S}_c)^i \bar{L}_c C E^{-1} (A E^{-1})^i \) approach zero. This allows us to approximate \( V \) and \( W \) by truncating these sums as follows:
\begin{align}
V&\approx\sum_{i=0}^{n_p}(E^{-1}A)^iE^{-1}B\bar{L}_b\bar{S}_b^i,\label{sum_dvt}\\
W^*&\approx\sum_{i=0}^{n_q}(\bar{S}_c)^i\bar{L}_cCE^{-1}(AE^{-1})^i.\label{sum_dwt}
\end{align} 
Next, define the following matrices:
\begin{align}
\tilde{V}&=\begin{bmatrix}E^{-1}B&\cdots&(E^{-1}A)^{n_p}E^{-1}B\end{bmatrix},\\
\hat{V}_r&=\begin{bmatrix}\bar{L}_b\\\vdots\\\bar{L}_b\bar{S}_b^{n_p}\end{bmatrix},\label{dtvr}\\
\tilde{W}^*&=\begin{bmatrix}CE^{-1}\\\vdots\\CE^{-1}(AE^{-1})^{n_q}\end{bmatrix},\\
\hat{W}_r^*&=\begin{bmatrix}\bar{L}_c&\cdots&(\bar{S}_c)^{n_q}\bar{L}_c\end{bmatrix}.\label{dtwr}
\end{align} From these definitions, it is evident that the sums (\ref{sum_dvt}) and (\ref{sum_dwt}) can be represented as \( \tilde{V} \hat{V}_r \) and \( \hat{W}_r^* \tilde{W}^* \), respectively. Thus, we have the approximations \( V \approx \tilde{V} \hat{V}_r \) and \( W \approx \tilde{W} \hat{W}_r \). Let us assume, for a moment, that this approximation is exact. In this case, the ROM satisfying the interpolation condition (\ref{int_cond_1}) can be obtained by reducing the impulse data quadruplet $(E_{k},A_{k},B_{k},C_{k})=(\tilde{W}^*E\tilde{V},\tilde{W}^*A\tilde{V},\tilde{W}^*B,C\tilde{V})$ as follows:
\begin{align}
E_r &= \hat{W}_r^*E_{k}\hat{V}_r, & A_r &= \hat{W}_r^*A_{k}\hat{V}_r, &B_r &= \hat{W}_r^*B_{k}, & C_r &= C_{k}\hat{V}_r. \label{k_interim_ROM}
\end{align}When $\sigma_j = \mu_i$, this ROM also satisfies the Hermite interpolation condition (\ref{int_cond_2}).

The impulse response of \( G(z) \) is given by  
\begin{align}
h(k)=C(E^{-1}A)^kE^{-1}B=CE^{-1}(AE^{-1})^kB.\nonumber
\end{align}
The impulse data quadruplet \( (E_{k}, A_{k}, B_{k}, C_{k}) \) is the same as the one used in the time-domain discrete-time QuadBT \cite{goseaQuad} and can be computed non-intrusively as follows:
\begin{align}
E_{k}&=\begin{bmatrix}h(0)&\cdots&h(n_p-1)\\\vdots&\ddots&\vdots\\h(n_q-1)&\cdots&h(n_p+n_q-2)\end{bmatrix},\nonumber\\
A_{k}&=\begin{bmatrix}h(1)&\cdots&h(n_p)\\\vdots&\ddots&\vdots\\h(n_q)&\cdots&h(n_p+n_q-1)\end{bmatrix},\nonumber\\
B_{k}&=\begin{bmatrix}h(0)\\\vdots\\h(n_q)\end{bmatrix},\quad C_{k}=\begin{bmatrix}h(0)&\cdots&h(n_p)\end{bmatrix}.\label{k_LF}
\end{align}Since $\hat{V}_r$ and $\hat{W}_r$ depend solely on the interpolation points $\sigma_j$ and $\mu_i$, and the tangential directions $b_j$ and $c_i$, the ROM $(E_r,A_r,B_r,C_r)$ can be computed non-intrusively.

It is now clear that DT-IRKA can be implemented using impulse response data \( h(k) \), eliminating the need for repeated estimations of \( G(\sigma_i) \) and \( G^\prime(\sigma_i) \) whenever DT-IRKA updates $\sigma_i$. The pseudo-code for the time-domain DT-IRKA (TD-DTIRKA) is provided in Algorithm \ref{alg5}.
\begin{algorithm}
\caption{TD-DTIRKA}
\textbf{Input:} Impulse response data: $\big(h(0)),\cdots,h(i_v)\big)$; Nodes: $(0,\cdots,i_v)$; Interpolation data: $(\sigma_1,\cdots,\sigma_r)$, $(b_1,\cdots,b_r)$, $(c_1,\cdots,c_r)$; Tolerance: tol.

\textbf{Output:} ROM: $(E_r,A_r,B_r,C_r)$
\begin{algorithmic}[1]\label{alg5}
\STATE Compute the impulse data quadruplet $(E_k,A_k,B_k,C_k)$ from (\ref{k_LF}).
\STATE \textbf{while}\Big(relative change in $\lambda_i$ > tol\Big)
\STATE Set $\bar{S}_b$, $\bar{L}_b$, $\bar{S}_c$, and $\bar{L}_c$ as in (\ref{dt_SbLbScLC}).
\STATE Set the projection matrices $\hat{V}_r$ and $\hat{W}_r$ as in (\ref{dtvr}) and (\ref{dtwr}).
\STATE Compute $(E_r,A_r,B_r,C_r)$ from (\ref{k_interim_ROM}).
\STATE Compute the eigenvalue decomposition: $E_r^{-1}A_r=T_r\Lambda T_r^{-1}$ where $\Lambda=diag(\lambda_1,\cdots,\lambda_r)$.
\STATE Update the interpolation data: $(\sigma_1,\cdots,\sigma_r)=(\frac{1}{\lambda_1},\cdots,\frac{1}{\lambda_r})$; $[b_1\cdots b_r]=B_r^*E_r^{-*}T_r^{-*}$; $[c_1^*\cdots c_r^*]=C_rT_r$.
\STATE \textbf{end while}
\end{algorithmic}
\end{algorithm}
\section{Pseudo-optimal Rational Krylov (PORK) Algorithm for Discrete-time Systems}\label{sec6}
In this section, we extend PORK to discrete-time systems and show that the discrete-time version maintains properties comparable to its continuous-time counterpart. Building on the findings from this section, we will formulate non-intrusive implementations of BT and DT-IRKA in the following section.
\subsection{Input PORK (I-PORK)}
By pre-multiplying equation (\ref{dsylv_V}) with \( W^* \), we obtain:
\begin{align}
A_r\bar{S}_b-E_r+B_r\bar{L}_b=0,\nonumber\\
A_r=(E_r-B_r\bar{L}_b)\bar{S}_b^{-1}.\nonumber
\end{align}This shows that \( A_r \) can be parameterized in terms of \( E_r \) and \( B_r \) without altering the interpolation conditions imposed by \( V \), as this is equivalent to varying \( W \).

Assume that the pair \( (\bar{S}_b, \bar{L}_b) \) is observable, and its observability Gramian \( \bar{Q}_s \) satisfies the following discrete-time Lyapunov equation:
\begin{align}
\bar{S}_b^*\bar{Q}_s\bar{S}_b-\bar{Q}_s+\bar{L}_b^*\bar{L}_b=0.\label{Qs}
\end{align}
\begin{theorem}By setting \( E_r = I \) and \( B_r = \bar{Q}_s^{-1}\bar{L}_b^* \), the following properties hold: \label{th1}
\begin{enumerate}
  \item $A_r=\bar{Q}_s^{-1}\bar{S}_b^*\bar{Q}_s$.
  \item The controllability Gramian \( P_r \) of the pair \( (A_r, B_r) \) is \( P_r = \bar{Q}_s^{-1} \).
  \item The ROM $(E_r,A_r,B_r,C_r)=(I,\bar{Q}_s^{-1}\bar{S}_b^*\bar{Q}_s,\bar{Q}_s^{-1}\bar{L}_b^*,CV)$ satisfies the optimality condition (\ref{opd3}).
\end{enumerate}
\end{theorem}
\begin{proof}Pre-multiplying (\ref{Qs}) by \( \bar{Q}_s^{-1} \) and post-multiplying by \( \bar{S}_b^{-1} \), we obtain:
\begin{align}
\bar{Q}_s^{-1}\bar{S}_b^*\bar{Q}_s-\big(I+\bar{Q}_s^{-1}\bar{L}_b^*\bar{L}_b\big)\bar{S}_b^{-1}=0.\nonumber
\end{align}Thus, \( A_r = \bar{Q}_s^{-1}\bar{S}_b^*\bar{Q}_s \).

The controllability Gramian \( P_r \) satisfies the discrete-time Lyapunov equation:
\begin{align}
A_rP_rA_r^T-E_rP_rE_r^T+B_rB_r^T&=0,\nonumber\\
\bar{Q}_s^{-1}\bar{S}_b^*\bar{Q}_sP_r\bar{Q}_s\bar{S}_b\bar{Q}_s^{-1}-P_r+\bar{Q}_s^{-1}\bar{L}_b^*\bar{L}_b\bar{Q}_s^{-1}&=0,\nonumber\\
\bar{S}_b^*\bar{Q}_sP_r\bar{Q}_s\bar{S}_b-\bar{Q}_sP_r\bar{Q}_s+\bar{L}_b^*\bar{L}_b&=0.\nonumber
\end{align}
Due to uniqueness, \( \bar{Q}_sP_r\bar{Q}_s = \bar{Q}_s \), and thus \( P_r = \bar{Q}_s^{-1} \).

Applying a state transformation using \( \bar{Q}_s \), the modal form of the ROM becomes:
\begin{align}
A_r&=\bar{S}_b^*,& B_r&=\bar{L}_b^*,& C_r&=C\bar{V}\bar{Q}_s^{-1}.\nonumber
\end{align}From the modal form, it is evident that this ROM satisfies the optimality condition $G\Big(\frac{1}{\hat{\lambda}_i^*}\Big)\hat{r}_i^*=G_r\Big(\frac{1}{\hat{\lambda}_i^*}\Big)\hat{r}_i^*$ since $\hat{\lambda}_i=\frac{1}{\sigma_i^*}$ and $\hat{r}_i=b_i^*$.
\end{proof}
\subsection{Output PORK (O-PORK)}
By taking the Hermitian of equation (\ref{dsylv_W}) and post-multiplying with \( V \), we obtain:
\begin{align}
\bar{S}_cA_r-E_r+\bar{L}_cC_r=0,\nonumber\\
A_r=\bar{S}_c^{-1}(E_r-\bar{L}_cC_r).\nonumber
\end{align}This demonstrates that \( A_r \) can be parameterized in terms of \( E_r \) and \( C_r \) without affecting the interpolation conditions imposed by \( W \), as this is equivalent to varying \( V \).

Assume that the pair \( (\bar{S}_c, \bar{L}_c) \) is controllable, and its controllability Gramian \( \bar{P}_s \) satisfies the following discrete-time Lyapunov equation:
\begin{align}
\bar{S}_c\bar{P}_s\bar{S}_c^*-\bar{P}_s+\bar{L}_c\bar{L}_c^*=0.\label{Ps}
\end{align}
\begin{theorem}By setting \( E_r = I \) and \( C_r = \bar{L}_c^*\bar{P}_s^{-1} \), the following properties hold:
\begin{enumerate}
  \item $A_r=\bar{P}_s\bar{S}_c^*\bar{P}_s^{-1}$.
  \item The observability Gramian \( Q_r \) of the pair \( (A_r, C_r) \) is \( Q_r = \bar{P}_s^{-1} \).
  \item The ROM $(E_r,A_r,B_r,C_r)=(I,\bar{P}_s\bar{S}_c^*\bar{P}_s^{-1},W^*B,\bar{L}_c^*\bar{P}_s^{-1})$ satisfies the optimality condition (\ref{opd2}).
\end{enumerate}
\end{theorem}
\begin{proof}
The proof is dual to that of Theorem \ref{th1} and is therefore omitted for brevity.
\end{proof}
\subsection{Approximation of Gramians}
Note that, similar to its continuous-time counterpart, PORK can be implemented non-intrusively using samples of \( G(z) \) at \( G(\sigma_i) \) and \( G(\mu_i) \) without any modifications. Additionally, discrete-time PORK also exhibits a monotonic decay in error as the number of interpolation points increases, analogous to its continuous-time version, as will be explained below.

Consider constructing an \((r-1)^{th}\)-order ROM \( G_{r-1}(z) \) using I-PORK with the right interpolation points \( (\sigma_1, \dots, \sigma_{r-1}) \) and tangential directions \( (b_1, \dots, b_{r-1}) \). Clearly, \( G_{r-1}(z) \), like \( G_r(z) \), satisfies the interpolation conditions for \( i = 1, \dots, r-1 \). Thus, \( G_{r-1}(z) \) is a pseudo-optimal ROM for both \( G_r(z) \) and \( G(z) \). Consequently, the following relationships hold:
\begin{align}
||G(z)-G_{r-1}(z)||_{\mathcal{H}_2}^2&=||G(z)||_{\mathcal{H}_2}^2-||G_{r-1}(z)||_{\mathcal{H}_2}^2,\nonumber\\
||G_r(z)-G_{r-1}(z)||_{\mathcal{H}_2}^2&=||G_r(z)||_{\mathcal{H}_2}^2-||G_{r-1}(z)||_{\mathcal{H}_2}^2,\nonumber\\
||G(z)-G_r(z)||_{\mathcal{H}_2}^2&=||G(z)||_{\mathcal{H}_2}^2-||G_r(z)||_{\mathcal{H}_2}^2,\nonumber\\
||G_r(z)||_{\mathcal{H}_2}^2&\geq||G_{r-1}(z)||_{\mathcal{H}_2}^2,\nonumber\\
||G(z)-G_r(z)||_{\mathcal{H}_2}^2&\leq||G(z)-G_{r-1}(z)||_{\mathcal{H}_2}^2.\nonumber
\end{align}Therefore, as the order of the ROM increases, \( ||G(z) - G_r(z)||_{\mathcal{H}_2} \) decays monotonically. A similar result can be shown for O-PORK.

Note that the controllability Gramian \( P \) and the observability Gramian \( Q \) of the discrete-time state-space realization \( (E, A, B, C) \) satisfy the following discrete-time Lyapunov equations:
\begin{align}
APA^T-EPE^T+BB^T=0,\nonumber\\
A^TQA-E^TQE+C^TC=0.\nonumber
\end{align}When either the optimality condition (\ref{opd2}) or (\ref{opd3}) is satisfied, the following holds:
\begin{align}
||G(z)-G_r(z)||_{\mathcal{H}_2}^2&=trace\big(C(P-VP_rV^*)C^T\big)=trace\big(B^T(Q-WQ_rW^*)B\big),\nonumber
\end{align}cf. \cite{bunse2010h2}. I-PORK can approximate \( P \) as \( P \approx VP_rV^* \), and O-PORK can approximate \( Q \) as \( Q \approx WQ_rW^* \). These approximations \( P \approx VP_rV^* \) and \( Q \approx WQ_rW^* \) monotonically approach \( P \) and \( Q \), respectively, as the number of interpolation points increases in PORK.
\section{Non-intrusive PORK-based Low-rank Balanced Truncation for Discrete Time Systems}\label{sec7}
The low-rank approximations of \( P \) and \( Q \) can be derived from the block version of discrete-time PORK, similar to the continuous-time case, by defining \( \bar{S}_b \), \( \bar{L}_b \), \( \bar{S}_c \), and \( \bar{L}_c \) as follows:
\begin{align}
\bar{S}_b&=\big(\text{diag}(\sigma_1,\cdots,\sigma_{n_p})\otimes I_m\big)^{-1},\nonumber\\
\bar{L}_b&=\Big(\begin{bmatrix}1&\cdots&1\end{bmatrix}\otimes I_m\Big)\bar{S}_b,\nonumber\\
\bar{S}_c&=\big(\text{diag}(\mu_1,\cdots,\mu_{n_q})\otimes I_p\big)^{-1},\nonumber\\
\bar{L}_c^*&=\Big(\begin{bmatrix}1&\cdots&1\end{bmatrix}\otimes I_p\Big)\bar{S}_c^*.\label{LbLc}
\end{align} The quality of approximation of $P$ and $Q$ can be tracked non-intrusively by observing the growth of $CV\bar{Q}_s^{-1}V^*C^T$ and $B^TW^*\bar{P}_s^{-1}W^*B$, respectively. Note that \( CV \), \( W^*B \), \( P_r = \bar{Q}_s^{-1} \), and \( Q_r = \bar{P}_s^{-1} \) can be computed using interpolation data and samples of \( G(z) \) at the interpolation points \( \sigma_i \) and \( \mu_i \). Furthermore, since \( W^*EV \) and \( W^*AV \) can also be computed non-intrusively from (\ref{s_LF}) via the Loewner framework, a non-intrusive low-rank BT algorithm can be formulated, analogous to its continuous-time counterpart. The pseudo-code for the non-intrusive PORK-based discrete-time BT (NI-PORK-DTBT) is presented in Algorithm \ref{alg6}.
\begin{algorithm}
\caption{NI-PORK-DTBT}
\textbf{Input:} Shifts for approximating $P$: $(\sigma_1,\cdots,\sigma_{n_p})$; Shifts for approximating $Q$: $(\mu_1,\cdots,\mu_{n_q})$; Frequency-domain data: $\big(G(\sigma_1),\cdots,G(\sigma_{n_p}),G(\mu_1),\cdots,G(\mu_{n_q})\big)$ and $G^\prime(\sigma_i)$ for $\sigma_i=\mu_j$; Reduced order: $r$.

\textbf{Output:} ROM: $(E_r,A_r,B_r,C_r)$
\begin{algorithmic}[1]\label{alg6}
\STATE Compute the Loewner quadruplet $(E_s,A_s,B_s,C_s)$ from (\ref{s_LF}).
\STATE Set $\bar{S}_b$, $\bar{S}_c$, $\bar{L}_b$, and $\bar{L}_c$ as in (\ref{LbLc}).
\STATE Compute $\bar{Q}_s$ and $\bar{P}_s$ by solving the discrete-time Lyapunov equations (\ref{Qs}) and (\ref{Ps}).
\STATE Decompose $\bar{Q}_s^{-1}=L_pL_p^*$ and $\bar{P}_s^{-1}=L_qL_q^*$.
\STATE Compute the projection matrices $\hat{V}_r$ and $\hat{W}_r$ from (\ref{proj_svd}) and (\ref{proj_mat}).
\STATE Compute the ROM from (\ref{dist_Es}).
\end{algorithmic}
\end{algorithm}

Similar to the continuous-time case, the computation of \( L_p \) and \( L_q \) can be done directly when \( \sigma_i \) and \( \mu_i \) are lightly damped. Let us express \( \sigma_i \) and \( \mu_i \) as follows:
\begin{align}
\sigma_i=e^{\frac{\zeta_{i,\sigma}|\omega_{i,\sigma}|}{\sqrt{1-\zeta_{i,\sigma}^2}}}e^{j\omega_{i,\sigma}}\quad \textnormal{and} \quad \mu_i=e^{\frac{\zeta_{i,\mu}|\omega_{i,\mu}|}{\sqrt{1-\zeta_{i,\mu}^2}}}e^{j\omega_{i,\mu}},\label{pork_shifts_disc}
\end{align}where $0<\zeta_{i,\sigma}\leq 1$ and $0<\zeta_{i,\mu}\leq 1$ are the damping coefficients of $\sigma_i$ and $\mu_i$. When \( \zeta_{i,\sigma} \ll 1 \) and \( \zeta_{i,\mu} \ll 1 \), $\bar{Q}_s^{-1}$, \( L_p \), $\bar{P}_s^{-1}$, and \( L_q \) approximate to
\begin{align}
\bar{Q}_s^{-1}&\approx\textnormal{diag}\big((\bar{\sigma}_1\sigma_1-1),\cdots,(\bar{\sigma}_{n_p}\sigma_{n_p}-1)\big)\otimes I_m,\label{dt_qs_formula}\\
L_p&\approx\textnormal{diag}\big(\sqrt{\bar{\sigma}_1\sigma_1-1},\cdots,\sqrt{\bar{\sigma}_{n_p}\sigma_{n_p}-1}\big)\otimes I_m,\\
\bar{P}_s^{-1}&\approx\textnormal{diag}\big((\bar{\mu}_1\mu_1-1),\cdots,(\bar{\mu}_{n_q}\mu_{n_q}-1)\big)\otimes I_p,\\
L_q&\approx\textnormal{diag}\big(\sqrt{\bar{\mu}_1\mu_1-1},\cdots,\sqrt{\bar{\mu}_{n_q}\mu_{n_q}-1}\big)\otimes I_p.\label{dt_lq_formula}
\end{align} The tangential version of PORK-based BT for discrete-time systems can be derived similarly to the continuous-time case, which is omitted here for brevity.
\section{Non-intrusive PORK-based DT-IRKA for Discrete Time Systems}\label{sec8}
Similar to the continuous-time case, a block PORK-based non-intrusive implementation of DT-IRKA can also be formulated. Here, the interpolation points \(\alpha_i\) and \(\beta_i\) are all located outside the unit circle. Let us define the following matrices:
\begin{align}
\bar{S}_\alpha&=S_\alpha^{-1},\quad \bar{L}_\alpha=L_\alpha S_\alpha^{-1},\quad\bar{S}_\beta=S_\beta^{-1},\quad \bar{L}_\beta=S_\beta^{-1}L_\beta.\label{dt_SaLaSbLb}
\end{align}
Let \(\bar{Q}_\alpha\) and \(\bar{P}_\beta\) be the solutions to the following discrete-time Lyapunov equations:
\begin{align}
\bar{S}_\alpha^*\bar{Q}_\alpha\bar{S}_\alpha-\bar{Q}_\alpha +\bar{L}_{\alpha}^*\bar{L}_\alpha&=0,\\
\bar{S}_\beta \bar{P}_\beta\bar{S}_\beta^*-\bar{P}_\beta+ \bar{L}_\beta \bar{L}_\beta^*&=0.
\end{align}
The ROM produced by discrete-time I-PORK is given by:
\begin{align}
E_\alpha&=I,&A_\alpha&=\bar{Q}_\alpha^{-1}\bar{S}_\alpha^*\bar{Q}_\alpha,\nonumber\\
B_\alpha&=\bar{Q}_\alpha^{-1}\bar{L}_\alpha^T,& \bar{C}_\alpha&=C\tilde{V}.
\end{align}
Similarly, the ROM produced by discrete-time O-PORK is given by:
\begin{align}
E_\beta&=I,&A_\beta&=\bar{P}_\beta \bar{S}_\beta^*\bar{P}_\beta^{-1},\nonumber\\
B_\beta&=\tilde{W}^*B,& C_\beta&=\bar{L}_\beta^*\bar{P}_\beta^{-1}.
\end{align}
Note that \(\bar{Q}_\alpha^{-1}\) and \(\bar{P}_\beta^{-1}\) can be computed directly from \(\alpha_i\) and \(\beta_i\) when they are lightly damped, as discussed in the previous subsection.

Let the projection matrices \(\hat{V}_r\) and \(\hat{W}_r\) be defined as:
\begin{align}
\hat{V}_r=\begin{bmatrix}(\sigma_1 I-A_\alpha)^{-1}B_\alpha b_1&\cdots&(\sigma_r I-A_\alpha)^{-1}B_\alpha b_r\end{bmatrix},\label{dt_Var}\\
\hat{W}_r=\begin{bmatrix}(\mu_1^*I-A_\beta^*)^{-1}C_\beta^*c_1^*&\cdots&(\mu_r^*I-A_\beta^*)^{-1}C_\beta^*c_r^*\end{bmatrix}.\label{dt_Wbr}
\end{align} It can then be observed that \(V \approx \tilde{V} \hat{V}_r\) and \(W \approx \tilde{W} \hat{W}_r\). Assuming this approximation is exact, the ROM satisfying the interpolation condition (\ref{int_cond_1}) can be obtained by reducing the Loewner quadruplet \((E_{\alpha,\beta}, A_{\alpha,\beta}, B_{\alpha,\beta}, C_{\alpha,\beta})\) as follows:
\begin{align}
E_r &= \hat{W}_r^*E_{\alpha,\beta}\hat{V}_r, & A_r &= \hat{W}_r^*A_{\alpha,\beta}\hat{V}_r, &B_r &= \hat{W}_r^*B_{\alpha,\beta}, & C_r &= C_{\alpha,\beta}\hat{V}_r. \label{dt_a_b_interim_ROM}
\end{align}
When \(\sigma_j = \mu_i\), this ROM also satisfies the Hermite interpolation condition (\ref{int_cond_2}). Since \(\hat{V}_r\) and \(\hat{W}_r\) depend solely on the interpolation points \(\alpha_j\), \(\beta_i\), \(\sigma_j\), and \(\mu_i\), as well as the tangential directions \(b_j\) and \(c_i\), the ROM \((E_r, A_r, B_r, C_r)\) can be computed in a non-intrusive manner.

It is now clear that DT-IRKA can be implemented using available transfer function samples \( G(\alpha_i) \) and \( G(\beta_i) \), eliminating the need for repeated estimations of \( G(\sigma_i) \) and \( G^\prime(\sigma_i) \) whenever DT-IRKA updates $\sigma_i$. The pseudo-code for the PORK-based DT-IRKA (PORK-DTIRKA) is outlined in Algorithm \ref{alg04}.
\begin{algorithm}
\caption{PORK-DTIRKA}
\textbf{Input:} Sampling points: $(\alpha_1,\cdots,\alpha_{n_p})$, $(\beta_1,\cdots,\beta_{n_q})$; Transfer function samples: $\big(G(\alpha_1),\cdots,G(\alpha_{n_v})\big)$, $\big(G(\beta_1),\cdots,G(\beta_{n_q})\big)$, $G^{\prime}(\alpha_i)$ for $\alpha_i=\beta_j$; Interpolation data: $(\sigma_1,\cdots,\sigma_r)$, $(b_1,\cdots,b_r)$, $(c_1,\cdots,c_r)$; Tolerance: tol.

\textbf{Output:} ROM: $(E_r,A_r,B_r,C_r)$
\begin{algorithmic}[1]\label{alg04}
\STATE Compute the Loewner quadruplet $(E_{\alpha,\beta},A_{\alpha,\beta},B_{\alpha,\beta},C_{\alpha,\beta})$ from (\ref{a_b_LF}).
\STATE \textbf{while}\Big(relative change in $\lambda_i$ > tol\Big)
\STATE Compute the projection matrices $\hat{V}_r$ and $\hat{W}_r$ from (\ref{dt_Var}) and (\ref{dt_Wbr}).
\STATE Compute $(E_r,A_r,B_r,C_r)$ from (\ref{dt_a_b_interim_ROM}).
\STATE Compute the eigenvalue decomposition: $E_r^{-1}A_r=T_r\Lambda T_r^{-1}$ where $\Lambda=diag(\lambda_1,\cdots,\lambda_r)$.
\STATE Update the interpolation data: $(\sigma_1,\cdots,\sigma_r)=(\frac{1}{\lambda_1},\cdots,\frac{1}{\lambda_r})$; $[b_1\cdots b_r]=B_r^*E_r^{-*}T_r^{-*}$; $[c_1^*\cdots c_r^*]=C_rT_r$.
\STATE \textbf{end while}
\end{algorithmic}
\end{algorithm}
\subsection{Tracking the Error $||G(z)-G_r(z)||_{\mathcal{H}_2}^2$}
Let \( G_r(z)^{(i-1)} \) and \( G_r(z)^{(i)} \) represent the interim ROMs in the \((i-1)^{th}\) and \(i^{th}\) iterations of DT-IRKA, respectively. Similar to the continuous-time case, the error in the \((i-1)^{th}\) iteration can be computed after the \(i^{th}\) iteration as follows:
\begin{align}
||G(z)-&G_r(z)^{(i-1)}||_{\mathcal{H}_2}^2\nonumber\\
=&||G(z)||_{\mathcal{H}_2}^2+||G_r(z)^{(i-1)}||_{\mathcal{H}_2}^2-2\text{trace}\Big(C_r^{(i)}\big(C_r^{(i-1)}T_r^{(i-1)}\big)^*\Big).\nonumber
\end{align}Thus, with a delay of one iteration, the error \( ||G(z) - G_r(z)||_{\mathcal{H}_2} \) can be tracked by computing \( ||G_r(z)||_{\mathcal{H}_2}^2 \) in each iteration. Specifically, the variable component of the error in non-intrusive DT-IRKA can be monitored non-intrusively by tracking the following term:
\begin{align}
||G_r(z)^{(i-1)}||_{\mathcal{H}_2}^2 - 2 \, \text{trace}\Big(C_r^{(i)}\big(C_r^{(i-1)}T_r^{(i-1)}\big)^*\Big). \nonumber
\end{align} However, it is important to note that the term \(-2 \, \text{trace}\Big(C_r^{(i)}\big(C_r^{(i-1)}T_r^{(i-1)}\big)^*\Big)\) is an approximation and not exact. Its accuracy depends on the precision of the approximation of the integral (\ref{dt_int_V}) and (\ref{dt_int_W}) or the approximation of the infinite summation (\ref{dt_sum_V}) and (\ref{dt_sum_W}).
\section{Compression and Distillation of Data Quadruplets}\label{sec9}
Throughout this paper, a consistent pattern has emerged in all the discussed non-intrusive algorithms: each algorithm constructs a Loewner quadruplet (in the frequency domain) or an impulse data quadruplet (in the time domain) and then reduces the respective data quadruplet, as illustrated in Figure \ref{fig0}.
\begin{figure}[!h]
  \centering
  \includegraphics[width=4cm]{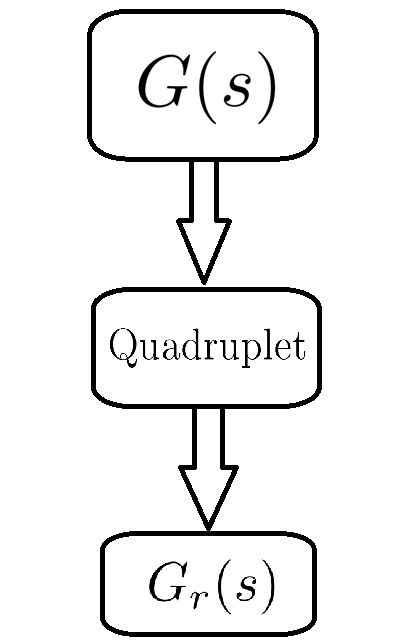}
  \caption{Working Principle}\label{fig0}
\end{figure}
It is now evident that all interpolatory low-rank BT algorithms, including Krylov-subspace-based low-rank BT, low-rank ADI-based BT, and QuadBT, construct the ROM by reducing the corresponding data quadruplets rather than directly reducing the original system. In intrusive settings, these quadruplets are not explicitly constructed, as the low-rank factor \(\hat{Z}_p\) is derived from the matrices \((E, A, B)\) separately, and the low-rank factor \(\hat{Z}_q\) is obtained from the matrices \((E, A, C)\) separately. In other words, the input and output dynamics are approximated independently. However, in non-intrusive settings, the true implicit nature of interpolatory low-rank BT algorithms becomes apparent, revealing that they construct the ROM by reducing the data quadruplets rather than directly reducing the original system.

Before proceeding further, let us make an assumption that \(mn_p = pn_q\), ensuring that the Loewner quadruplets are interpolants of \(G(s)\). This assumption will greatly simplify our discussion, as it allows us to use the terms Loewner quadruplet and interpolant of \(G(s)\) interchangeably. Consequently, we can analyze the Loewner quadruplet using standard interpolation theory.

The following observations can be made regarding interpolatory low-rank BT methods:
\begin{enumerate}
  \item Similar to numerical integration, interpolatory low-rank BT does not reduce \(G(s)\) directly. Instead, an interpolant of \(G(s)\) is first implicitly constructed (or explicitly constructed in non-intrusive settings). This interpolant is not particularly compact, as it is constructed to interpolate \(G(s)\) at several interpolation points to capture the majority of the original system's dynamics. Subsequently, this interpolant acts as a surrogate for \(G(s)\). The ROMs produced by these low-rank BT algorithms are approximations of the interpolants of \(G(s)\), rather than \(G(s)\) itself. In this sense, low-rank BT could be termed ``numerical BT'' if we wish to adopt terminology analogous to numerical integration.
  \item QuadBT and the non-intrusive BT algorithms proposed in this paper are non-intrusive but they perform intrusive MOR on the interpolant of \(G(s)\), for which a state-space realization can be conveniently obtained non-intrusively within the Loewner framework.
  \item In \cite{goseaQuad}, QuadBT was contrasted with the interpolatory Loewner framework, stressing that its associated matrices are diagonally scaled Loewner matrices rather than standard Loewner matrices. It was also asserted that rational interpolation has no role in QuadBT. However, in Subsection \ref{subsec_QuadBT}, our presentation of QuadBT deliberately avoided multiplying the Loewner matrices by $L_p$ and $L_q$ to emphasize that these matrices simply perform a similarity transformation on the Loewner quadruplet without altering its fundamental interpolatory properties. This slight rearrangement of variables reveals that rational interpolation plays a key role in QuadBT, as it supplies QuadBT with a state-space realization constructed non-intrusively from data. This realization is then further reduced using the balancing square-root algorithm, which uses $L_p$ and $L_q$ as similarity transformations to obtain a suitable realization of the interpolant before reducing it. The final ROM is not an interpolant, but that does not mean rational interpolation plays no role in QuadBT or that the Loewner quadruplet differs from the one appearing in interpolation. QuadBT effectively further reduces the same interpolant that appears in the interpolatory Loewner framework.  
  \item Since the ROMs produced by interpolatory low-rank BT are approximations of the interpolants of \(G(s)\), it is unreasonable to expect that reducing the order of the interpolant will result in a final ROM that is more accurate than the interpolant itself. Therefore, the accuracy of the approximation in low-rank BT is directly tied to the quality of the interpolant of \(G(s)\). To ensure that low-rank BT generates ROMs nearly equivalent to those produced by standard BT, the interpolation quality must be exceptional, which heavily relies on the selection of interpolation points. Given that IRKA is regarded as one of the most effective interpolation algorithms, its ROMs should be considered strong candidates for performing low-rank BT. This is supported by \cite{benner2014self}, where IRKA is used to generate effective shifts for the ADI method.
  \item There is some interest within the MOR community to produce BT models through interpolation; see \cite{ionescu2012balancing,kawano2023gramian}. These efforts are primarily focused on constructing exact BT models using interpolation techniques. However, it is important to recognize that, in an approximate sense, low-rank BT algorithms are already producing BT models via interpolation. When we acknowledge the success of ADI-based or Krylov-subspace-based algorithms in extending the applicability of BT to large-scale systems by reducing computational costs, we are indirectly affirming that interpolation at a small number of points may not surpass BT in accuracy. However, if interpolation is performed more liberally, it can achieve sufficient accuracy to compete with BT. Interpolation at a large number of points, while powerful, introduces its own complexities, which will be discussed shortly. Nevertheless, the accuracy and effectiveness of interpolation as a tool in MOR must be acknowledged.
  \item The non-intrusive IRKA algorithms presented in this paper leverage the same principles as low-rank BT. They compute an interpolant of \(G(s)\) by interpolating at several points to capture the majority of the dynamics of \(G(s)\). This interpolant then serves as a surrogate for \(G(s)\), allowing the algorithms to sample the interpolant as IRKA updates the interpolation points, rather than directly sampling \(G(s)\). This approach enables the non-intrusive IRKA algorithms to bypass the need for new experiments to obtain additional samples of \(G(s)\).
\end{enumerate}
Having established that all interpolatory low-rank BT algorithms essentially reduce their respective data quadruplets, one might consider directly applying standard MOR algorithms like BT and IRKA to these quadruplets to obtain a compact ROM. However, these quadruplets are often not as well-behaved as desired. In many cases, when constructing an interpolant in the Loewner framework with a large number of interpolation points, the resulting interpolant is an unstable system with several poles in the right-half plane \cite{mayo2007framework,mao2024data}. As a result, standard MOR algorithms that require a stable original model cannot be directly applied to reduce the size of these quadruplets. Additionally, the Loewner matrix \(W^T E V\) tends to become singular as the number of interpolation points increases \cite{mayo2007framework,mao2024data}, rendering MOR algorithms that assume the non-singularity of the \(E\)-matrix unsuitable for directly reducing the order of Loewner interpolants. QuadBT and the algorithms proposed in this paper can be viewed as ``compression'' algorithms, designed to extract a compact ROM from these quadruplets. Moreover, these algorithms can also be seen as ``distillation'' algorithms, as they can extract ROMs with various properties from the same ``raw'' quadruplet by processing it differently. They effectively distill a compact, useful, and well-behaved ROM from the raw data quadruplets, which cannot be directly handled by standard MOR algorithms that assume the original model is well-behaved (like stable and minimal).
\section{Numerical Examples}\label{sec10}
This section evaluates the performance of the proposed non-intrusive algorithms by comparing them to their intrusive counterparts. The effectiveness of these algorithms is illustrated through four numerical examples: the first two involve continuous-time systems, while the remaining two focus on discrete-time cases.
\subsection{Experimental Setup}\label{subex_1}
For quadrature-based algorithms, QuadBT \cite{goseaQuad} is first used to generate ROMs. The proposed quadrature-based IRKA algorithms then compress and distill the same quadruplets produced by QuadBT to target a local optimum of $||G(s)-G_r(s)||_{\mathcal{H}_2}^2$. To ensure a fair comparison between NI-ADI-BT and frequency-domain QuadBT, the ADI shifts are chosen using (\ref{adi_shifts_cont}) with \(\zeta_{i,\sigma} = \zeta_{i,\mu} = 10^{-4}\), so that both methods compress nearly identical quadruplets. Likewise, NI-ADI-BT and PORK-IRKA compress the same quadruplet, with \(L_p\) and \(L_q\) in NI-ADI-BT, and \(Q_s^{-1}\) and \(P_s^{-1}\) in PORK-IRKA, computed via (\ref{qs_formula})–(\ref{lq_formula}).

For discrete-time systems, the procedure is similar: QuadBT \cite{goseaQuad} generates initial ROMs, and the proposed quadrature-based DT-IRKA algorithms compress and distill the same quadruplet to target a local optimum of $||G(z)-G_r(z)||_{\mathcal{H}_2}^2$. To ensure a fair comparison between discrete-time QuadBT and NI-PORK-DTBT, interpolation points are selected using (\ref{pork_shifts_disc}) with \(\zeta_{i,\sigma} = \zeta_{i,\mu} = 10^{-4}\), ensuring both methods act on nearly identical quadruplets. In this case, \(L_p\) and \(L_q\) in NI-PORK-DTBT, and \(\bar{Q}_s^{-1}\) and \(\bar{P}_s^{-1}\) in PORK-DTIRKA, are obtained from (\ref{dt_qs_formula})–(\ref{dt_lq_formula}).

All IRKA and DT-IRKA-based algorithms are initialized arbitrarily and stopped after 50 iterations if convergence is not achieved. All experiments are conducted in MATLAB R2021b on a Windows 11 laptop with a 2GHz Intel Core i7 processor and 16GB of RAM. The MATLAB codes for reproducing the results in this section are provided in \cite{mycode}. 
\subsection{CD Player}
The CD Player model is a \(120^{\text{th}}\)-order system with 2 inputs and 2 outputs, taken from the benchmark collection in \cite{chahlaoui2005benchmark} commonly used for evaluating MOR algorithms. For frequency-domain QuadBT, 300 logarithmically-spaced quadrature nodes and weights are generated over the frequency range \(10^{-3}\) to \(10^3\) rad/sec using the exponential trapezoidal rule, which is the preferred numerical quadrature method in \cite{goseaQuad} for achieving high accuracy. These nodes are used to approximate both the controllability and observability Gramians. The corresponding transfer function samples are computed from the state-space realization of the CD Player model provided in \cite{chahlaoui2005benchmark}. For time-limited QuadBT, 400 uniformly-spaced quadrature nodes and weights are computed over the time interval \([0, 40]\) seconds using the Gauss–Legendre quadrature rule. Impulse response samples are generated using the same state-space model. The respective quadruplets are then constructed and used to execute QuadBT. Subsequently, 300 ADI shifts are generated as described in Subsection \ref{subex_1}, after which transfer function samples are evaluated and the corresponding quadruplet is constructed.

Figure~\ref{fig1} shows the largest 20 Hankel singular values approximated by QuadBT and NI-ADI-BT.
\begin{figure}[!h]
  \centering
  \includegraphics[width=10cm]{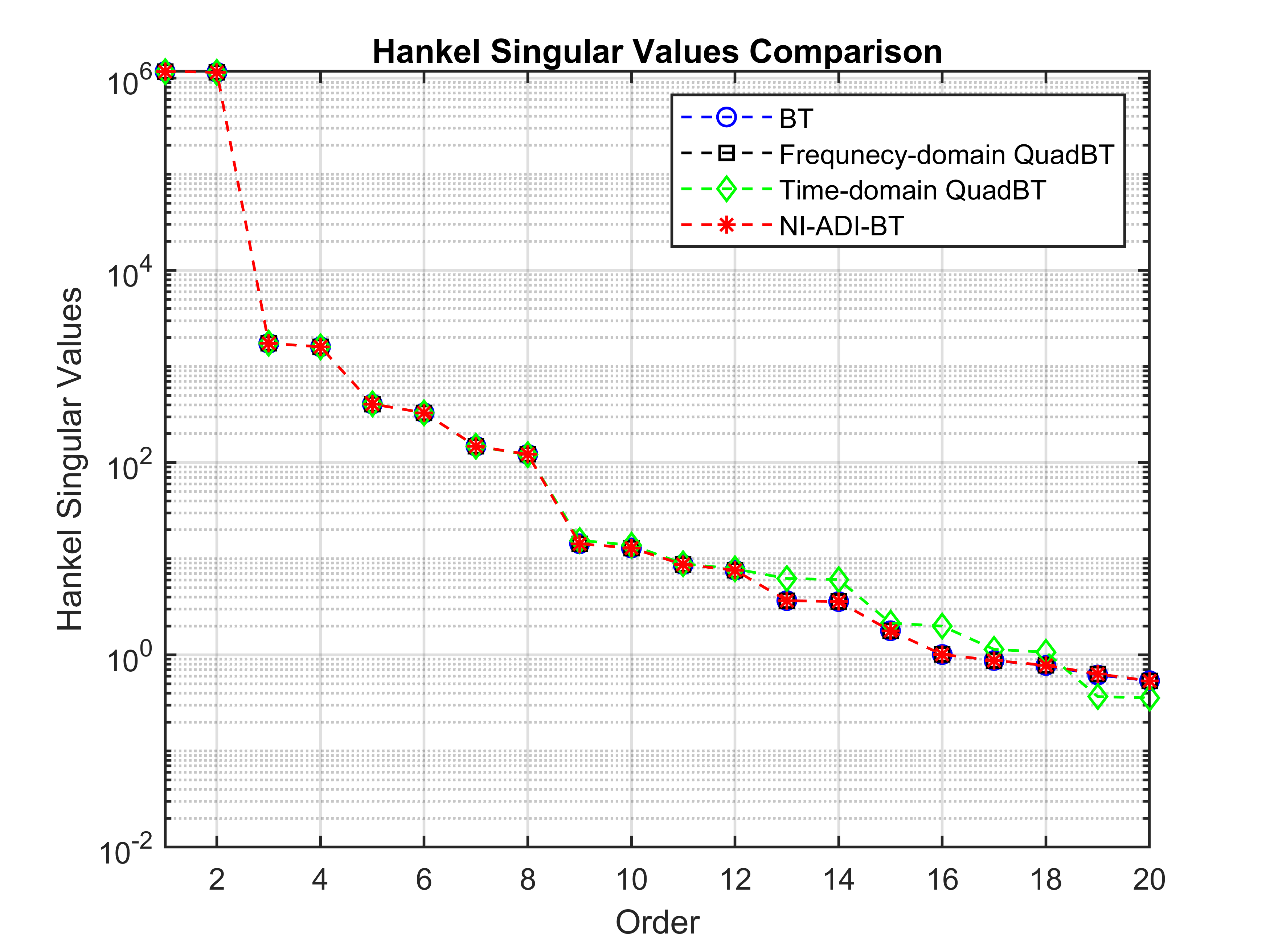}
  \caption{Comparison of Hankel singular values}\label{fig1}
\end{figure} As shown, NI-ADI-BT provides a close approximation to the Hankel singular values of the original system. Figure~\ref{fig2} presents the \(\mathcal{H}_\infty\) norm of the relative error \(\frac{||G(s) - G_r(s)||_{\mathcal{H}_\infty}}{||G(s)||_{\mathcal{H}_\infty}}\) for ROMs of orders 1 through 20. The results indicate that NI-ADI-BT achieves accuracy comparable to that of intrusive BT and QuadBT.
\begin{figure}[!h]
  \centering
  \includegraphics[width=10cm]{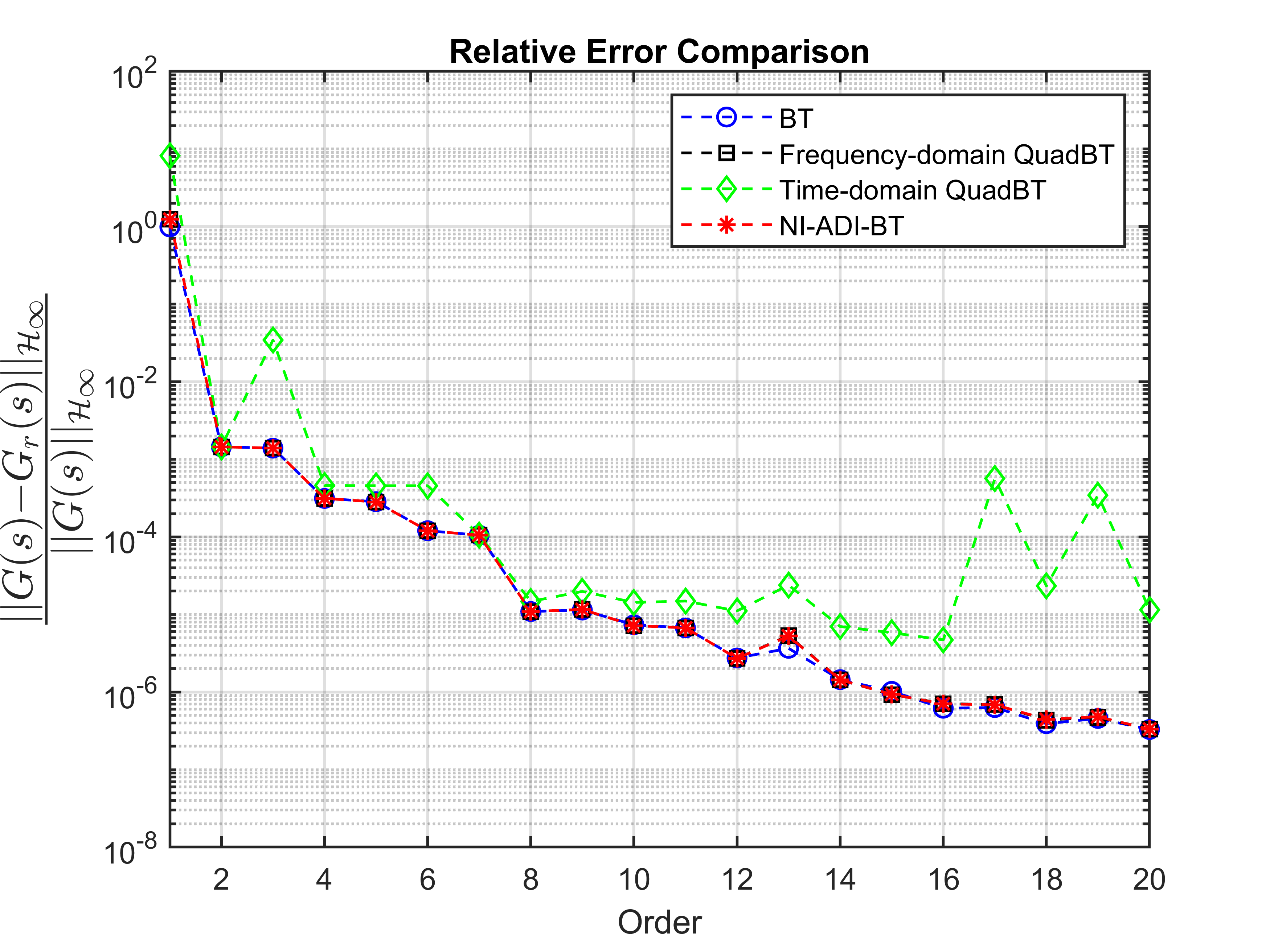}
  \caption{Comparison of relative error $\frac{||G(s)-G_r(s)||_{\mathcal{H}_\infty}}{||G(s)||_{\mathcal{H}_\infty}}$}\label{fig2}
\end{figure}

Using the same respective quadruplets, FD-Quad-IRKA, TD-Quad-IRKA, and PORK-IRKA are applied to compute an $8^{th}$-order ROM. In FD-Quad-IRKA, the weights are obtained using the trapezoidal rule applied to the same nodes used by QuadBT. Figure~\ref{fig3} shows the frequency response of the original system \(G(s)\) (input 1 and output 1) and the ROMs \(G_r(s)\) (input 1 and output 1) produced by IRKA, FD-Quad-IRKA, TD-Quad-IRKA, and PORK-IRKA. The results show that the proposed quadrature-based IRKA methods offer accuracy comparable to that of IRKA. For brevity, only the frequency response of the first input-output channel is shown.
\begin{figure}[!h]
  \centering
  \includegraphics[width=10cm]{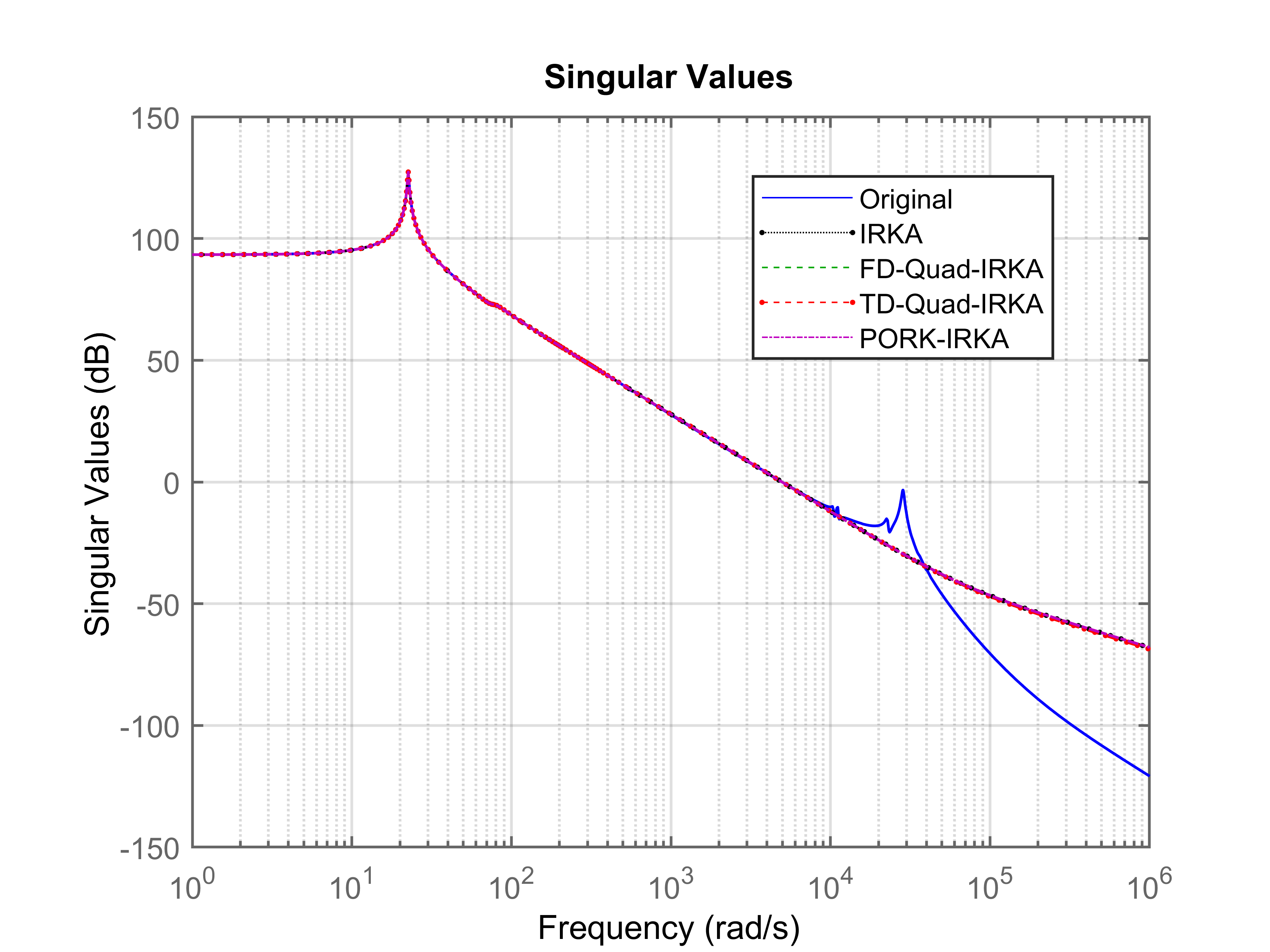}
  \caption{Frequency responses of $G(s)$ and $G_r(s)$}\label{fig3}
\end{figure}
\subsection{International Space Station}
The International Space Station model is a \(270^{\text{th}}\)-order system with 3 inputs and 3 outputs, taken from the benchmark collection in \cite{chahlaoui2005benchmark}. For frequency-domain QuadBT, 400 logarithmically-spaced quadrature nodes and weights are generated over the frequency range \(10^{-1}\) to \(10^2\) rad/sec using the exponential trapezoidal rule. These are used to approximate both the controllability and observability Gramians. The corresponding transfer function samples are computed from the state-space realization of the model provided in \cite{chahlaoui2005benchmark}. For time-limited QuadBT, 400 uniformly-spaced quadrature nodes and weights are computed over the interval \([0, 40]\) seconds using the Gauss–Legendre quadrature rule. Impulse response samples are generated using the same state-space model. The resulting quadruplets are then constructed and used to run QuadBT. Afterward, 400 ADI shifts are generated as described in Subsection \ref{subex_1}, transfer function samples are evaluated, and the associated quadruplet is constructed.

Figure~\ref{fig4} shows the 30 largest Hankel singular values approximated by QuadBT and NI-ADI-BT.
\begin{figure}[!h]
  \centering
  \includegraphics[width=10cm]{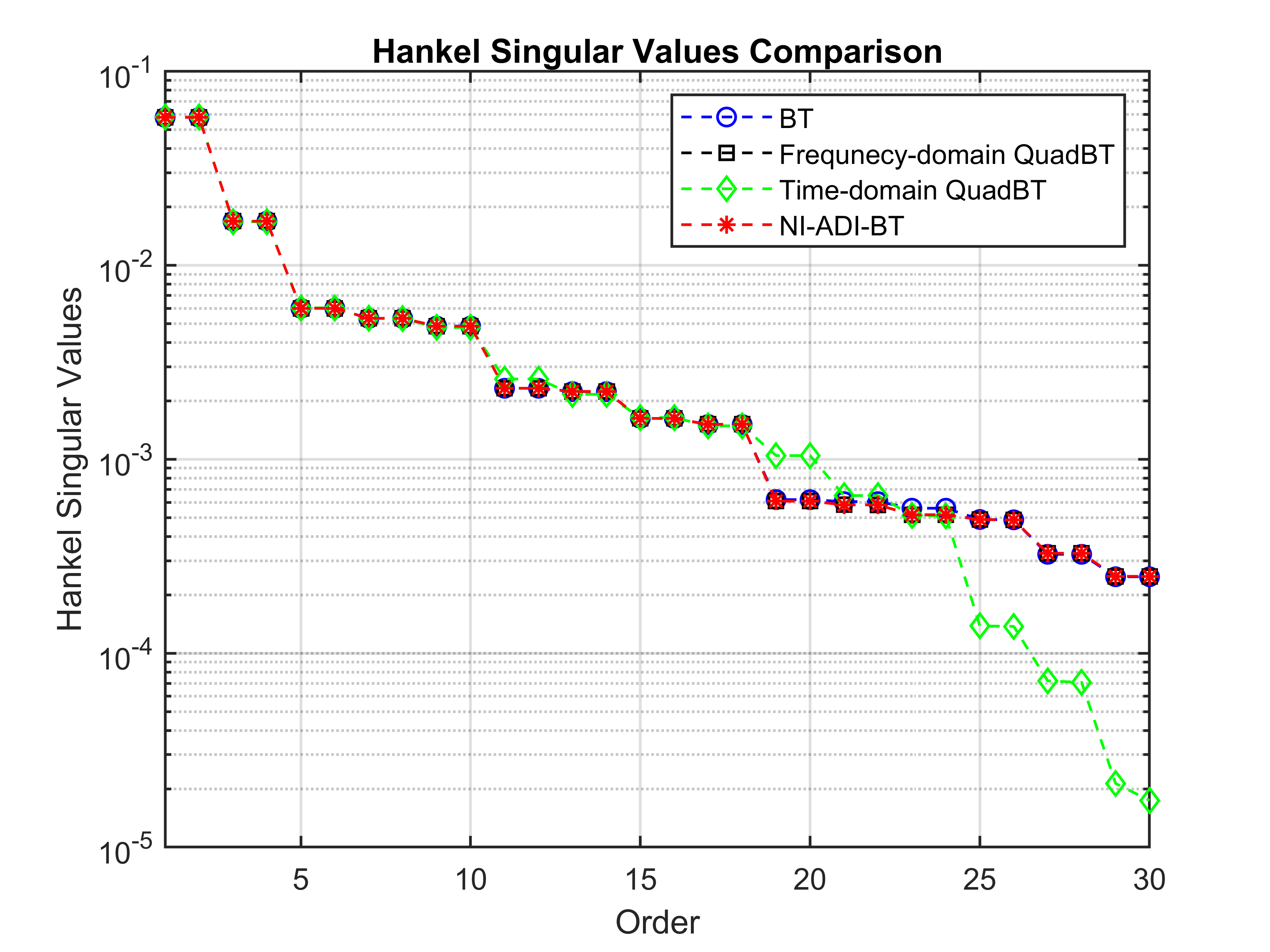}
  \caption{Comparison of Hankel singular values}\label{fig4}
\end{figure} As shown, NI-ADI-BT closely matches the Hankel singular values of the original system. Figure~\ref{fig5} presents the \(\mathcal{H}_\infty\) norm of the relative error \(\frac{||G(s) - G_r(s)||_{\mathcal{H}_\infty}}{||G(s)||_{\mathcal{H}_\infty}}\) for ROMs of orders 1 through 30. The results demonstrate that NI-ADI-BT attains accuracy comparable to that of QuadBT.
\begin{figure}[!h]
  \centering
  \includegraphics[width=10cm]{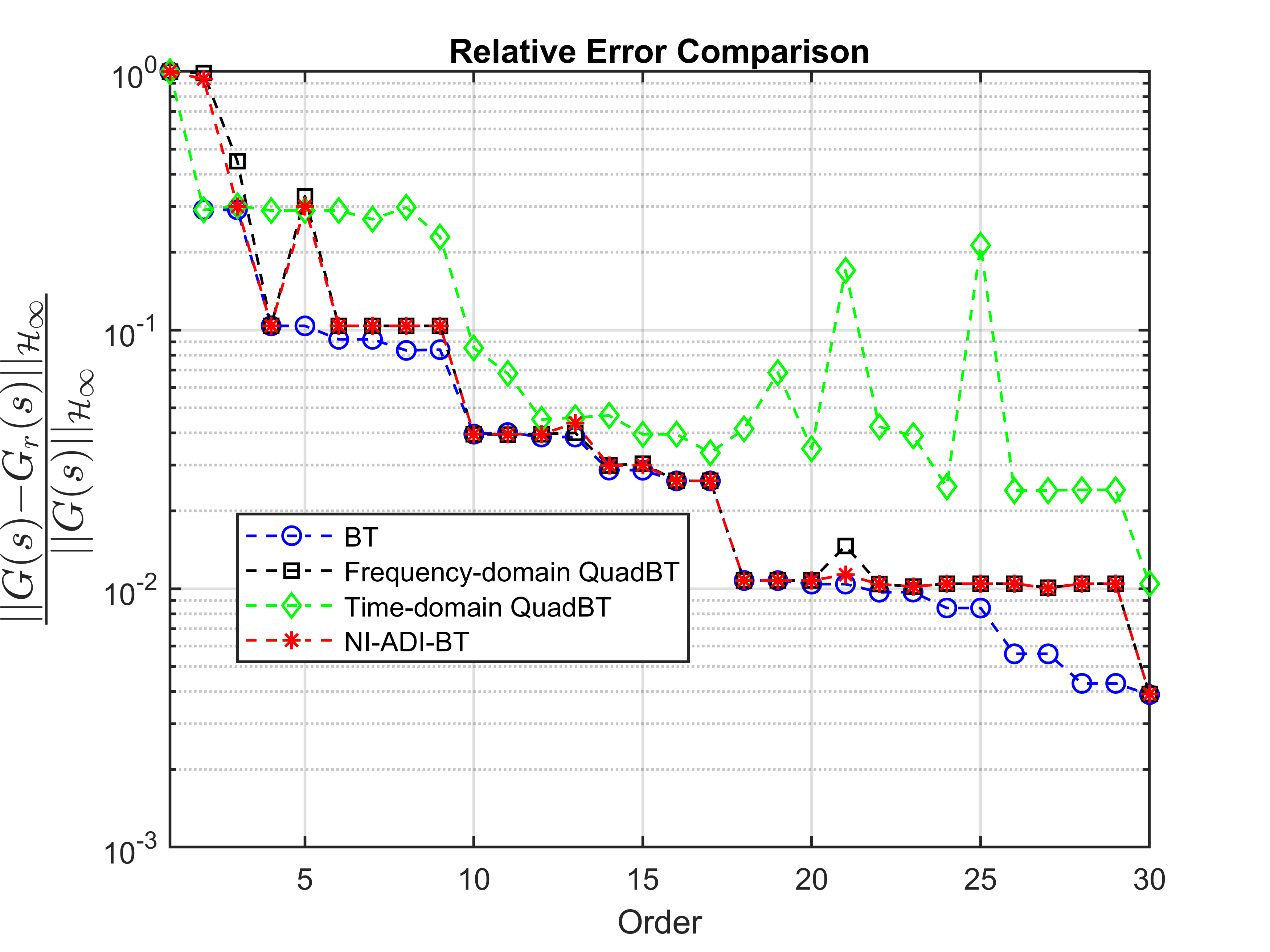}
  \caption{Comparison of relative error $\frac{||G(s)-G_r(s)||_{\mathcal{H}_\infty}}{||G(s)||_{\mathcal{H}_\infty}}$}\label{fig5}
\end{figure}

Using the same respective quadruplets, FD-Quad-IRKA, TD-Quad-IRKA, and PORK-IRKA are used to compute a \(12^{\text{th}}\)-order ROM. In FD-Quad-IRKA, the weights are computed using the trapezoidal rule applied to the same nodes as those used by QuadBT. Figure~\ref{fig6} shows the frequency response of the original system \(G(s)\) (input 1 and output 1) and the ROMs \(G_r(s)\) (input 1 and output 1) produced by IRKA, FD-Quad-IRKA, TD-Quad-IRKA, and PORK-IRKA. The plots demonstrate that the proposed quadrature-based IRKA methods achieve accuracy comparable to that of IRKA. For brevity, only the first input-output channel is shown.
\begin{figure}[!h]
  \centering
  \includegraphics[width=10cm]{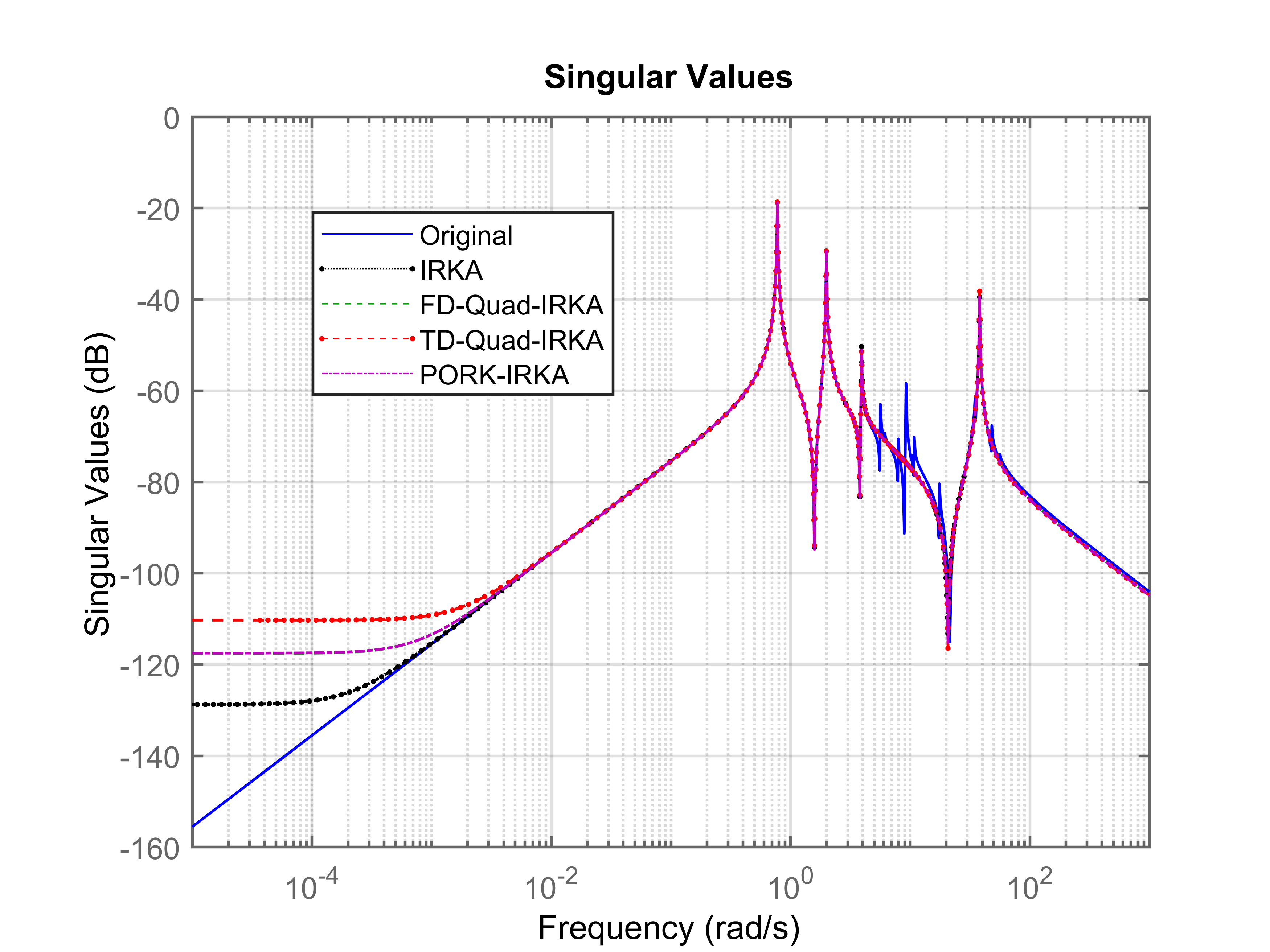}
  \caption{Frequency responses of $G(s)$ and $G_r(s)$}\label{fig6}
\end{figure}
\subsection{Low-pass Butterworth Filter}
In this example, we consider the model from \cite{goseaQuad}, which is a \(40^{\text{th}}\)-order low-pass Butterworth filter with a cutoff frequency of \(0.6\) rad/sec. For frequency-domain QuadBT, 100 uniformly-spaced Gauss–Legendre quadrature nodes and weights are generated over the frequency range \(-\pi\) to \(\pi\) rad/sec. These are used to approximate both the controllability and observability Gramians. Transfer function samples at the quadrature nodes are obtained from the state-space realization of the Butterworth filter, generated using MATLAB’s `\textit{butter}` command. For time-limited QuadBT, 100 impulse response samples are generated from the same state-space model. The corresponding quadruplets are then constructed and used to perform QuadBT. Next, 100 PORK shifts are computed as described in Subsection \ref{subex_1}, followed by the evaluation of transfer function samples and construction of the associated quadruplet.

Figure~\ref{fig7} presents the 20 largest Hankel singular values approximated by QuadBT and NI-PORK-DTBT.
\begin{figure}[!h]
  \centering
  \includegraphics[width=10cm]{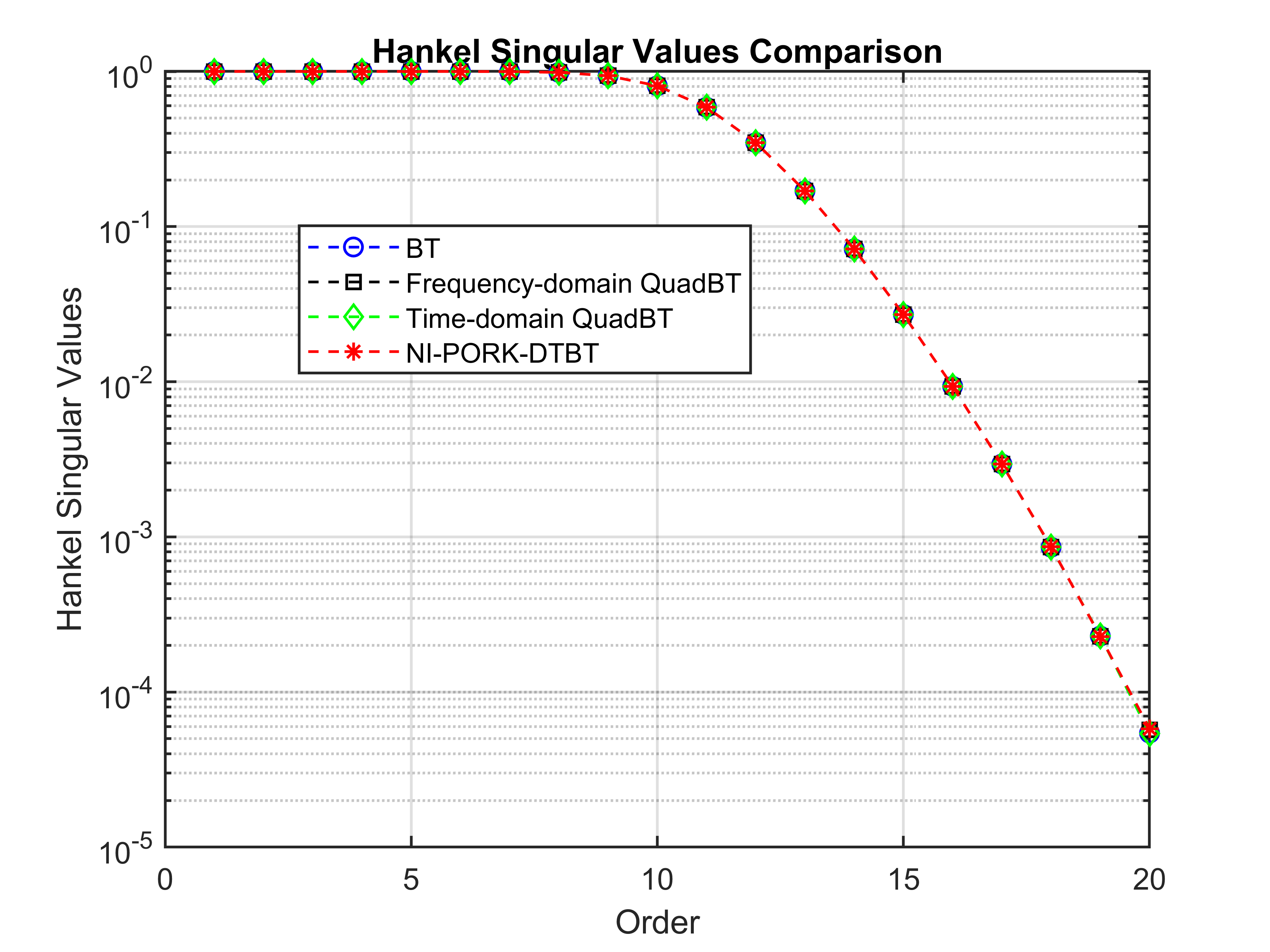}
  \caption{Comparison of Hankel singular values}\label{fig7}
\end{figure} As shown, NI-PORK-DTBT closely matches all 20 Hankel singular values. Figure~\ref{fig8} shows the \(\mathcal{H}_\infty\) norm of the relative error \(\frac{||G(z) - G_r(z)||_{\mathcal{H}_\infty}}{||G(z)||_{\mathcal{H}_\infty}}\) for ROMs of orders 1 through 20. The results indicate that NI-PORK-DTBT provides accuracy comparable to that of QuadBT.
\begin{figure}[!h]
  \centering
  \includegraphics[width=10cm]{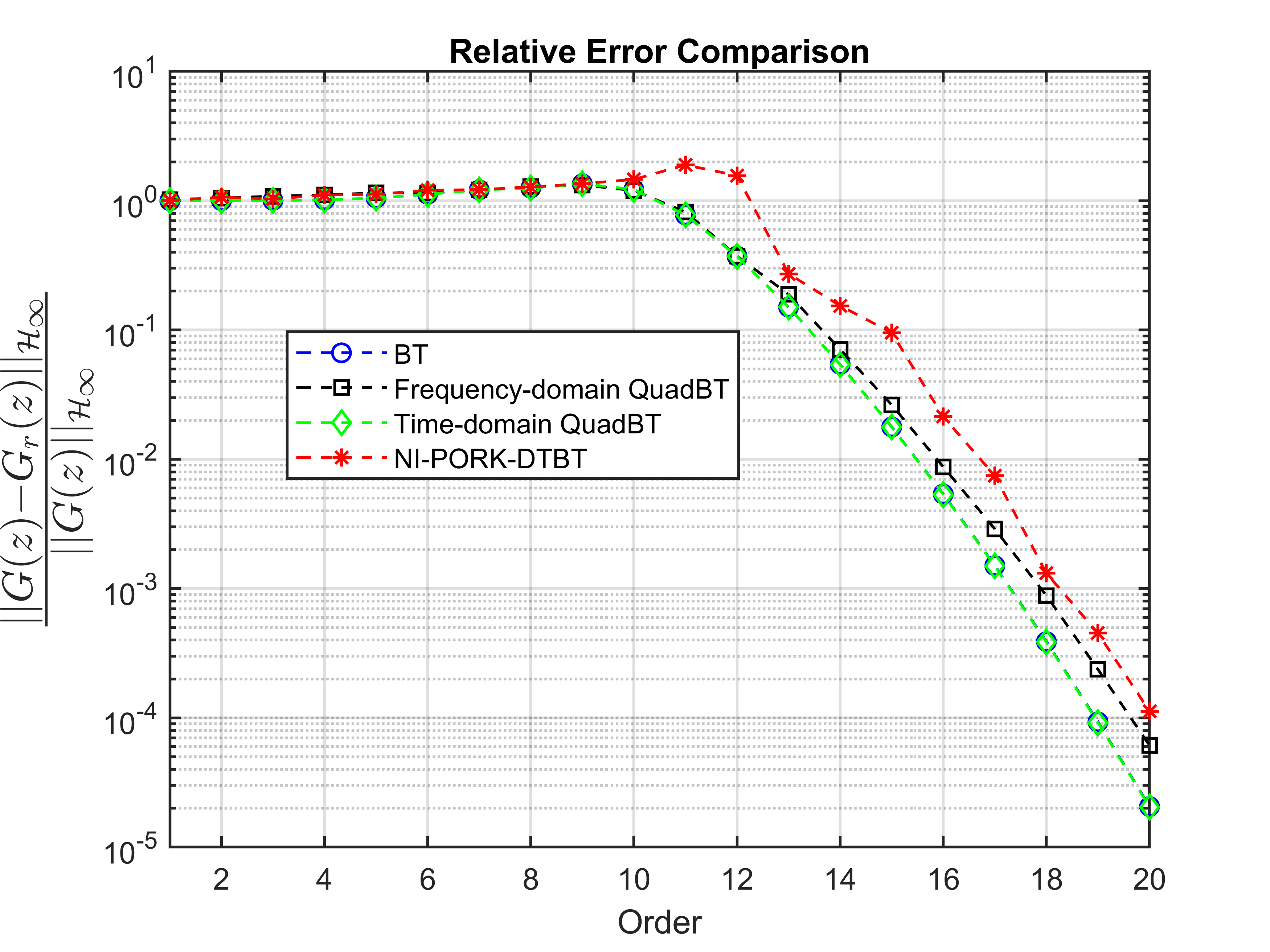}
  \caption{Comparison of relative error $\frac{||G(z)-G_r(z)||_{\mathcal{H}_\infty}}{||G(z)||_{\mathcal{H}_\infty}}$}\label{fig8}
\end{figure}

Using the same quadruplets, FD-Quad-DTIRKA, TD-DTIRKA, and PORK-DTIRKA are applied to compute a \(15^{\text{th}}\)-order ROM. Figure~\ref{fig9} displays the frequency response of \(G(z)\) and the ROMs \(G_r(z)\) produced by DT-IRKA, FD-Quad-DTIRKA, TD-DTIRKA, and PORK-DTIRKA. The plots show that the proposed methods achieve accuracy on par with DT-IRKA.
\begin{figure}[!h]
  \centering
  \includegraphics[width=10cm]{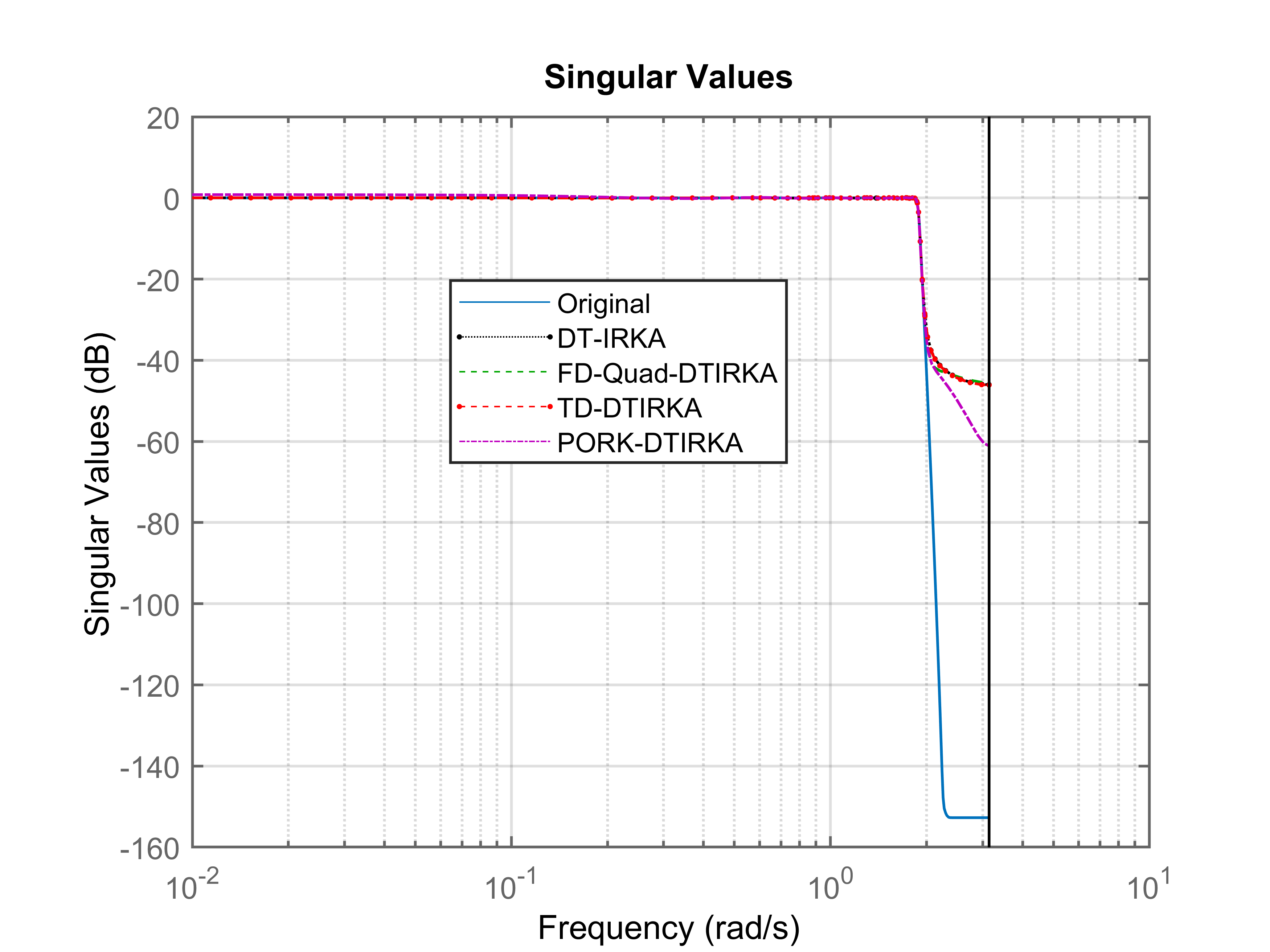}
  \caption{Frequency responses of $G(z)$ and $G_r(z)$}\label{fig9}
\end{figure}
\subsection{Heat Transfer Equation}
The Heat Transfer Equation model is a \(200^{\text{th}}\)-order system, taken from the benchmark collection in \cite{chahlaoui2005benchmark}. For frequency-domain QuadBT, 100 uniformly-spaced Gauss–Legendre quadrature nodes and weights are generated over the frequency range \(-\pi\) to \(\pi\) rad/sec. These are used to approximate both the controllability and observability Gramians. Transfer function samples at the quadrature nodes are obtained from the state-space realization provided in \cite{chahlaoui2005benchmark}. For time-limited QuadBT, 100 impulse response samples are generated from the same state-space model. The corresponding quadruplets are then constructed and used to perform QuadBT. Next, 100 PORK shifts are computed as described in Subsection \ref{subex_1}, followed by the evaluation of transfer function samples and construction of the associated quadruplet.

Figure~\ref{fig10} presents the 30 largest Hankel singular values approximated by QuadBT and NI-PORK-DTBT.
\begin{figure}[!h]
  \centering
  \includegraphics[width=10cm]{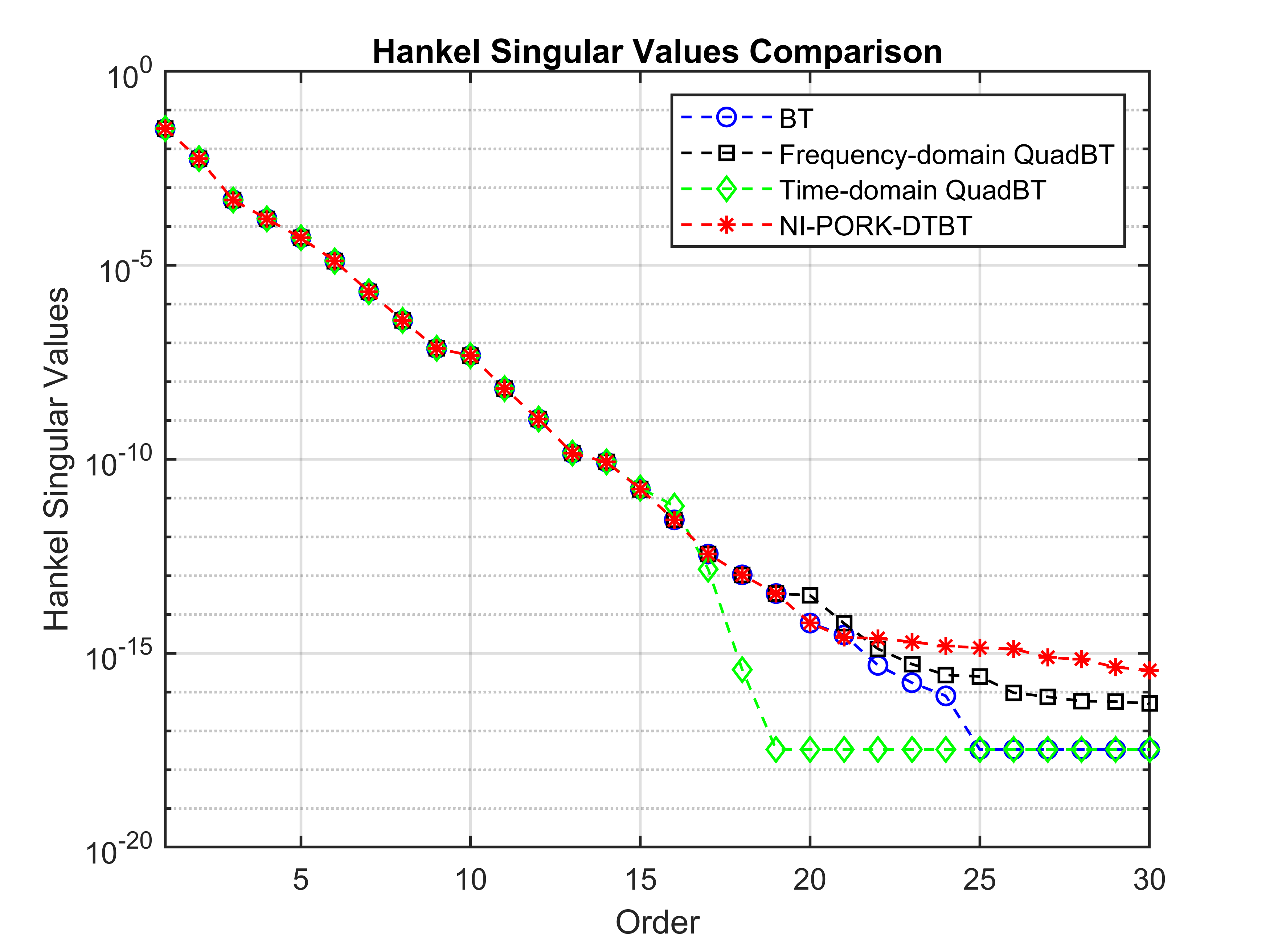}
  \caption{Comparison of Hankel singular values}\label{fig10}
\end{figure} As shown, NI-PORK-DTBT closely most of the Hankel singular values. Figure~\ref{fig11} shows the \(\mathcal{H}_\infty\) norm of the relative error \(\frac{||G(z) - G_r(z)||_{\mathcal{H}_\infty}}{||G(z)||_{\mathcal{H}_\infty}}\) for ROMs of orders 1 through 30. The results indicate that NI-PORK-DTBT provides accuracy comparable to that of QuadBT.
\begin{figure}[!h]
  \centering
  \includegraphics[width=10cm]{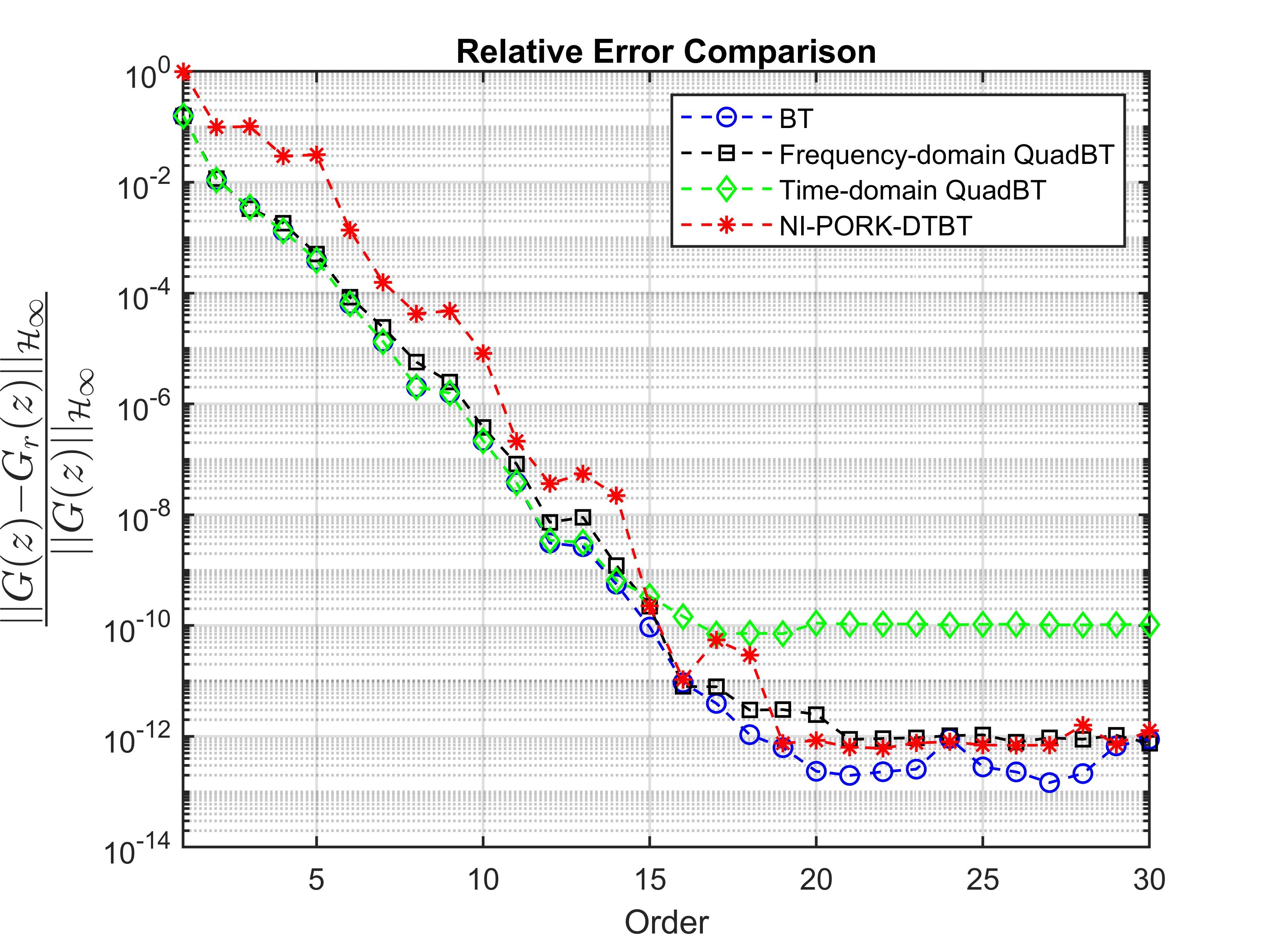}
  \caption{Comparison of relative error $\frac{||G(z)-G_r(z)||_{\mathcal{H}_\infty}}{||G(z)||_{\mathcal{H}_\infty}}$}\label{fig11}
\end{figure}

Using the same quadruplets, FD-Quad-DTIRKA, TD-DTIRKA, and PORK-DTIRKA are applied to compute a \(10^{\text{th}}\)-order ROM. Figure~\ref{fig12} displays the frequency response of \(G(z)\) and the ROMs \(G_r(z)\) produced by DT-IRKA, FD-Quad-DTIRKA, TD-DTIRKA, and PORK-DTIRKA. The plots show that the proposed methods achieve accuracy on par with DT-IRKA.
\begin{figure}[!h]
  \centering
  \includegraphics[width=10cm]{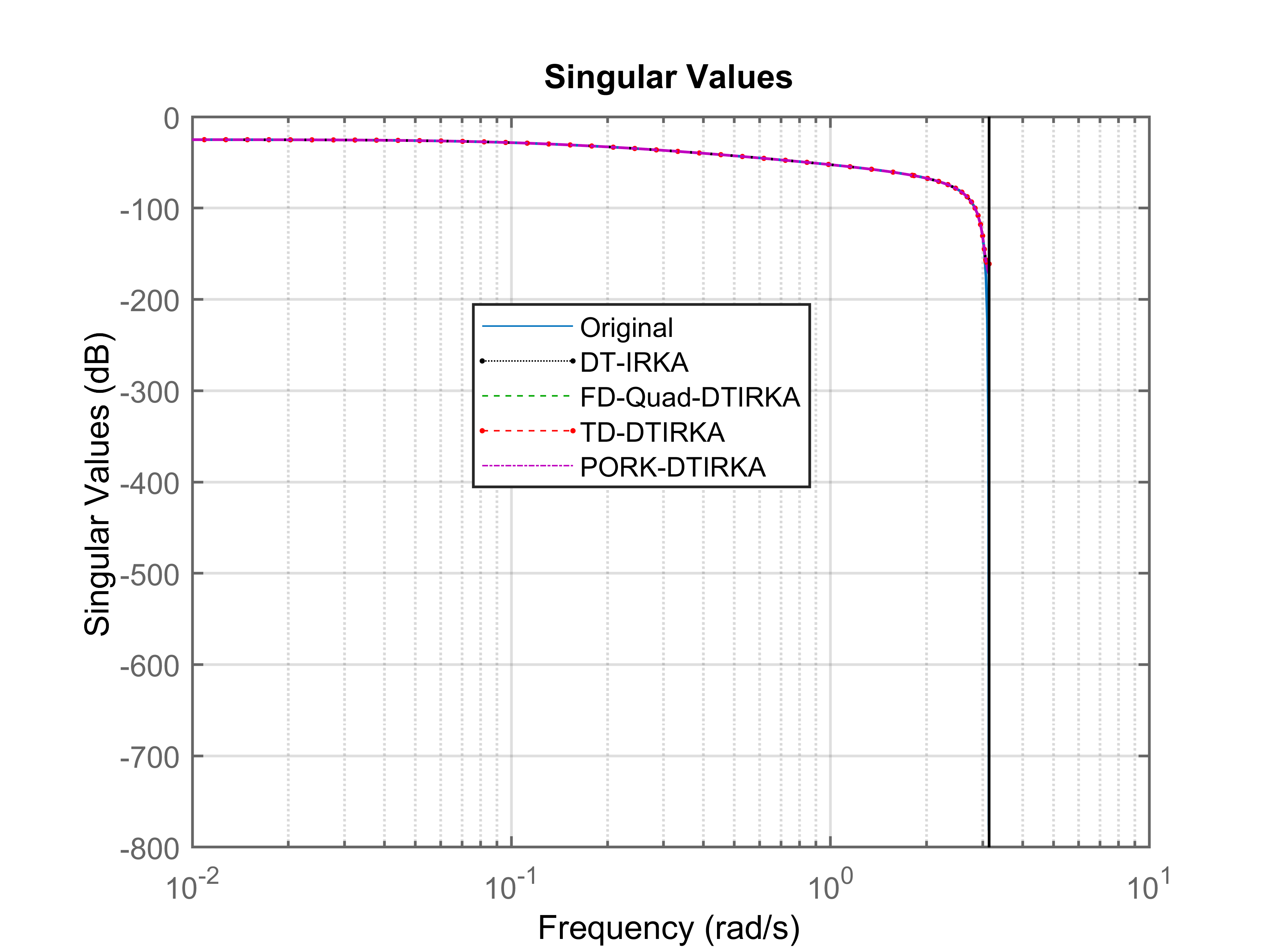}
  \caption{Frequency responses of $G(z)$ and $G_r(z)$}\label{fig12}
\end{figure}
\section{Conclusion}\label{sec11}
This paper presents non-intrusive implementations of BT and IRKA for both continuous-time and discrete-time systems. The proposed methods utilize available frequency or time-domain data to compute ROMs. It has been observed that both QuadBT and the algorithms introduced in this paper effectively compress and distill their respective raw quadruplets, resulting in compact and practical ROMs. Numerical experiments demonstrate that the proposed algorithms perform comparably to their intrusive counterparts.
\section*{Acknowledgements and Funding}                               
We are deeply grateful to Ion Victor Gosea at the Max Planck Institute for Dynamics of Complex Technical Systems in Magdeburg, Germany, for his patient responses to our numerous questions about the Loewner framework and for his valuable feedback that helped improve our manuscript. This work was supported by the National Natural Science Foundation of China under Grants No. 62350410484, 62273059, 62376146, and 62176001, and in part by the Shanxi Province Central Guidance for Local Science and Technology Development Special Project (Grant No. YDZJSX20231D003).

\end{document}